\documentclass[a4paper, authoryear, 12pt]{elsarticle}

\usepackage[english]{babel}
\usepackage{babel}
\usepackage{fontenc}
\usepackage{graphicx}
\usepackage{amsmath,amsthm,amssymb}
\usepackage{natbib}
\usepackage[colorlinks=true,citecolor=blue]{hyperref}
\usepackage{color}
\usepackage{graphicx}
\usepackage{subfigure}
\usepackage{color}
\usepackage{newlfont}
\usepackage{multirow}
\usepackage{longtable} 
\usepackage{ifthen}
\usepackage{alltt}
\usepackage{enumerate}
\usepackage[text={6.25in,8.5in},centering]{geometry}
\usepackage{tikz}
\usepackage{epstopdf}
\usepackage{float}
\usepackage{arydshln}
\usepackage[utf8]{inputenc}
\usepackage{tikz}
\usetikzlibrary{shapes,arrows}
\numberwithin{equation}{section}
\newtheorem{theorem}{Theorem}[section]
\newtheorem{lemma}{Lemma}[section]
\newtheorem{proposition}{\bf Proposition}[section]

\input epsf.sty
\usepackage{lineno}
\journal{arXiv}

\begin{document}
	\title{Short-term predictions and prevention strategies for COVID-19: A model-based study}
	
	\author[ISI]{Sk Shahid Nadim}
	\author[ISI]{Indrajit Ghosh \footnote{Corresponding author. Email: indra7math@gmail.com, indrajitg\_r@isical.ac.in}}
	\author[ISI]{Joydev Chattopadhyay}
	\address[ISI]{Agricultural and Ecological Research Unit, Indian Statistical Institute, Kolkata - 700 108, West Bengal, India}

\begin{abstract}
	An outbreak of respiratory disease caused by a novel coronavirus is ongoing from December 2019. As of July 22, 2020, it has caused an epidemic outbreak with more than 15 million confirmed infections and above 6 hundred thousand reported deaths worldwide. During this period of an epidemic when human-to-human transmission is established and reported cases of coronavirus disease 2019 (COVID-19) are rising worldwide, investigation of control strategies and forecasting are necessary for health care planning. In this study, we propose and analyze a compartmental epidemic model of COVID-19 to predict and control the outbreak. The basic reproduction number and control reproduction number are calculated analytically. A detailed stability analysis of the model is performed to observe the dynamics of the system. We calibrated the proposed model to fit daily data from the United Kingdom (UK) where the situation is still alarming. Our findings suggest that independent self-sustaining human-to-human spread ($R_0>1$, $R_c>1$) is already present. Short-term predictions show that the decreasing trend of new COVID-19 cases is well captured by the model. Further, we found that effective management of quarantined individuals is more effective than management of isolated individuals to reduce the disease burden. Thus, if limited resources are available, then investing on the quarantined individuals will be more fruitful in terms of reduction of cases. 
\end{abstract}


\maketitle

\section{Introduction}
In December 2019, an outbreak of novel coronavirus (2019-nCoV) infection, was first noted in Wuhan, Central China \cite{Who2019}. The outbreak was declared a public health emergency of international concern on 30 January 2020 by WHO. Coronaviruses belong to the Coronaviridae family and widely distributed in humans and other mammals \cite{huang2020clinical}. The virus is responsible for a range of symptoms including dry cough, fever, fatigue, breathing difficulty, and bilateral lung infiltration in severe cases, similar to those caused by SARS-CoV and MERS-CoV infections \cite{huang2020clinical, gralinski2020return}. Many people may experience non-breathing symptoms including nausea, vomiting and diarrhea \cite{cdcgov2020}. Some patients have reported radiographic changes in their ground-glass lungs; normal or lower than average white blood  cell lymphocyte, and platelet counts; hypoxaemia; and deranged liver and renal function. Most of them were said to be geographically connected to the Huanan seafood wholesale market, which was subsequently claimed by journalists to be selling freshly slaughtered game animals \cite{chinadaily2019}. The Chinese health authority said the patients initially tested negative for common respiratory viruses and bacteria but subsequently tested positive for a novel coronavirus (nCoV) \cite{chan2020familial}. In contrast to the initial findings  \cite{cheng20202019}, the 2019-nCoV virus spreads from person to person as confirmed in \cite{chan2020familial}. It has become an epidemic outbreak with more than 15 million confirmed infections and above 6 hundred thousand deaths worldwide as of 22 July 2020. The current epidemic outbreak result in 2,85,768 confirmed cases and 44,236 deaths in the UK \cite{Worldometer2020}. Since first discovery and identification of coronavirus in 1965, three major outbreaks occurred, caused by emerging, highly pathogenic coronaviruses, namely the 2003 outbreak of Severe Acute Respiratory Syndrome (SARS) in mainland China \cite{gumel2004modelling,li2003angiotensin}, the 2012 outbreak of Middle East Respiratory Syndrome (MERS) in Saudi Arabia \cite{de2013commentary,sardar2020realistic}, and the 2015 outbreak of MERS in South Korea \cite{cowling2015preliminary,kim2017middle}. These outbreaks resulted in SARS and MERS cases confirmed by more than 8000 and 2200, respectively \cite{kwok2019epidemic}. The COVID-19 is caused by a new  genetically similar corona virus to the viruses that cause SARS and MERS. Despite a relatively lower death rate compared to SARS and MERS, the COVID-19 spreads rapidly and infects more people than the SARS and MERS outbreaks. In spite of strict intervention measures implemented in the region where the infection originated, the infection spread locally in Wuhan, in China and around the globally. 

On 31 January 2020, the UK reported the first confirmed case of acute respiratory infection due to corona virus disease 2019 (COVID-19), and initially responded to the spread of infection by quarantining at-risk individuals. As of 28 June 2020, there were 3,12,654 confirmed cases and 43,730 confirmed cases deaths, the world's second highest per capita death rate among the major nations \cite{Worldometer2020}. Within the hospitals the infection rate is higher than in the population. In March 23, the UK government implemented a lock-down and declared that everyone should start social distancing immediately, suggesting that contact with others will be avoided as far as possible. Entire households should also quarantine themselves for 14 days if anyone has a symptom of COVID-19, and anyone at high risk of serious illness should isolate themselves for 12 weeks, including pregnant women, people over 70 and those with other health conditions. The country is literally at a standstill and the disease has seriously impacted the economy and the livelihood of the people.

As the 2019 coronavirus disease outbreak (COVID-19) is expanding rapidly in UK, real-time analyzes of epidemiological data are required to increase situational awareness and inform interventions. Earlier, in the first few weeks of an outbreak, real-time analysis shed light on the severity, transmissibility, and natural history of an emerging pathogen, such as SARS, the 2009 influenza pandemic, and Ebola \cite{chowell2009severe,chowell2011characterizing,fraser2009pandemic,lipsitch2003transmission}. Analysis of detailed patient line lists is especially useful for inferring key epidemiological parameters, such as infectious and incubation periods, and delays between infection and detection, isolation and case reporting \cite{chowell2009severe,chowell2011characterizing}. However, official patient's health data seldom become available to the public early in an outbreak, when the information is most required. In addition to medical and biological research, theoretical studies based on either mathematical or statistical modeling may also play an important role throughout this anti-epidemic fight in understanding the epidemic character traits of the outbreak, in predicting the inflection point and end time, and in having to decide on the measures to reduce the spread. To this end, many efforts have been made at the early stage to estimate key epidemic parameters and forecast future cases in which the statistical models are mostly used \cite{muniz2020epidemic,lai2020assessing,chakraborty2020real}. An Imperial College London study group calculated that 4000 (95\% CI: 1000-9700) cases had occurred in Wuhan with symptoms beginning on January 18, 2020, and an estimated basic reproduction number was 2.6 (95\% CI: 1.5-3.5) using the number of cases transported from Wuhan to other countries \cite{imai2020estimating}. Leung et al. reached a similar finding, calculating the number of cases transported from Wuhan to other major cities in China \cite{nowcast2019} and also suggesting the possibility for the spreading of risk \cite{bogoch2020pneumonia} for travel-related diseases. Mathematical modeling based on dynamic equations \cite{tang2020updated,sardar2020assessment,kucharski2020early,aldila2020mathematical,rajagopal2020fractional,britton2020mathematical} may provide detailed mechanism for the disease dynamics. Several studies were based on the UK COVID-19 situation \cite{davies2020effects,clark2020global,adameffectiveness,jit2020estimating}. Davies et. al \cite{davies2020effects} studied the potential impact of different control measures for mitigating the burden of COVID-19 in the UK. They used a stochastic age-structured transmission model to explore a range of intervention scenarios. These studies has broadly suggested that control measures could reduce the burden of COVID-19. However, there is a scope of comparing popular intervention strategies namely, quarantine and isolation utilizing recent epidemic data from the UK.  

In this study, we aim to study the control strategies that can significantly reduce the outbreak using a mathematical modeling framework. By mathematical analysis of the proposed model we would like to explore transmission dynamics of the virus among humans. Another goal is the short-term prediction of new COVID-19 cases in the UK.

\section{Model formulation}{\label{Model_formulation}}
General mathematical models for the spread of infectious diseases have been described previously \cite{may1991infectious,diekmann2000mathematical,hethcote2000mathematics}.
A compartmental differential equation model for COVID-19 is formulated and analyzed. We adopt a variant that reflects some key epidemiological properties of COVID-19. The model monitors the dynamics of seven sub-populations, namely susceptible $(S(t))$, exposed $(E(t))$, quarantined $(Q(t))$, asymptomatic $(A(t))$,  symptomatic $(I(t))$, isolated $(J(t))$ and recovered $(R(t))$ individuals. The total population size is $N(t)= S(t) + E(t) + Q(t) + A(t)+I(t)+J(t)+ R(t)$. In this model, quarantine refers to the separation of COVID-19 infected individuals from the general population when the population are infected but not infectious, whereas isolation describes the separation of COVID-19 infected individuals when the population become symptomatic infectious. Our model incorporates some demographic effects by assuming a proportional natural death rate $\mu>0$ in each of the seven sub-populations of the model. In addition, our model includes a net inflow of susceptible individuals into the region at a rate $\Pi$ per unit time. This parameter includes new births, immigration and emigration. The flow diagram of the proposed model is displayed in Figure \ref{Fig:flow_chart}.\\\\
\begin{figure}[ht]
    \centering
	\includegraphics[width=0.9\textwidth]{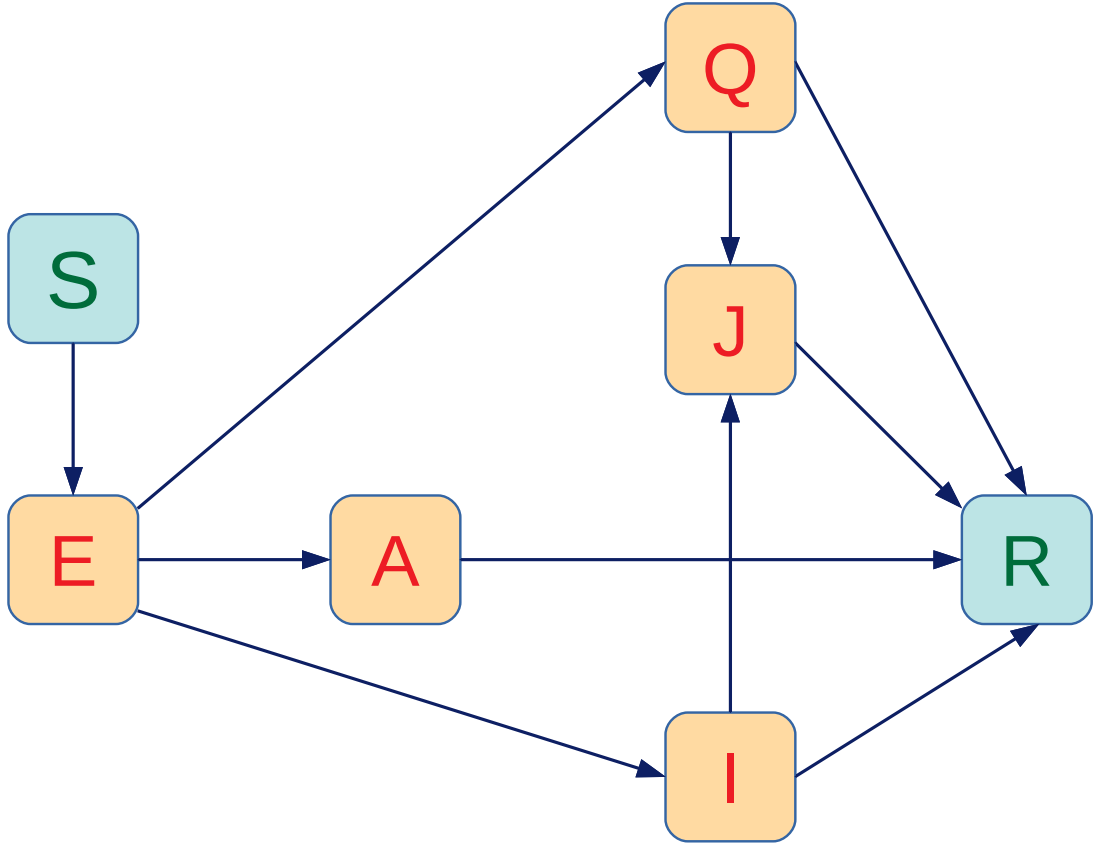}
	\caption{Compartmental flow diagram of the proposed model.}
	\label{Fig:flow_chart}
\end{figure}

\textbf{Susceptible population (S(t)):}

By recruiting individuals into the region, the susceptible population is increased and reduced by natural death. Also the susceptible population decreases after infection, acquired through interaction between a susceptible individual and an infected person who may be quarantined, asymptomatic, symptomatic, or isolated. For these four groups of infected individuals, the transmission coefficients are $\beta$, $r_Q \beta$, $r_A \beta$, and $r_J \beta$ respectively. We consider the $\beta$ as a transmission rate along with the modification factors for quarantined $r_Q$, asymptomatic $r_A$ and isolated $r_J$ individuals. The interaction between infected individuals (quarantined, asymptomatic, symptomatic or isolated) and susceptible is modelled in the form of total population without quarantined and isolated individuals using standard mixing incidence \cite{may1991infectious,diekmann2000mathematical,hethcote2000mathematics}. The rate of change of the susceptible population can be expressed by the following equation:
\begin{eqnarray}\label{EQ:eqn 2.01}
\displaystyle{\frac{dS}{dt}} &=& \Pi-\frac{S(\beta I+r_Q \beta Q+r_A \beta A+r_J \beta J)}{N}-\mu S,
\end{eqnarray}
\textbf{Exposed population(E(t)):}

Population who are exposed are infected individuals but not infectious for the community. The exposed population decreases with quarantine at a rate of $\gamma_1$, and become asymptomatic and symptomatic at a rate $k_1$ and natural death at a rate $\mu$. Hence,
\begin{eqnarray}
\displaystyle{\frac{dE}{dt}} &=& \frac{S(\beta I+r_Q \beta Q+r_A \beta A+r_J \beta J)}{N}-(\gamma_1+k_1+\mu)E
\end{eqnarray}
\textbf{Quarantine population (Q(t)):}

These are exposed individuals who are quarantined at a rate $\gamma_1$. For convenience, we consider that all quarantined individuals are exposed who will begin to develop symptoms and then transfer to the isolated class. Assuming that a certain portion of uninfected individuals are also quarantined would be more plausible, but this would drastically complicate the model and require the introduction of many parameters and compartments. In addition, the error caused by our simplification is to leave certain people in the susceptible population who are currently in quarantine and therefore make less contacts. The population is reduced by growth of clinical symptom at a rate of $k_2$ and transferred to the isolated class. $\sigma_1$ is the recovery rate of quarantine individuals and $\mu$ is the natural death rate of human population. Thus,
\begin{eqnarray}
\displaystyle{\frac{dQ}{dt}} &=& \gamma_1 E-(k_2+\sigma_1+\mu)Q
\end{eqnarray}
\textbf{Asymptomatic population(A(t)):}

Asymptomatic individuals were exposed to the virus but clinical signs of COVID have not yet developed. The exposed individuals become asymptomatic at a rate $k_1$ by a proportion $p$. The recovery rate of asymptomatic individuals is $\sigma_2$  and the natural death rate is $\mu$. Thus,
\begin{eqnarray}
\displaystyle{\frac{dA}{dt}} &=&  p k_1 E-(\sigma_2+\mu)A
\end{eqnarray}
\textbf{Symptomatic population(I(t)):}

The symptomatic individuals are produced by a proportion of $(1-p)$ of exposed class after the exposer of clinical symptoms of COVID by exposed individuals. $\gamma_2$ is the isolation rate of the symptomatic individuals, $\sigma_3$ is the recovery rate and natural death at a rate $\mu$. Thus,
\begin{eqnarray}
\displaystyle{\frac{dI}{dt}} &=& (1-p)k_1 E-(\gamma_2+\sigma_3+\mu)I
\end{eqnarray}
\textbf{Isolated population(J(t)):}

The isolated individuals are those who have been developed by clinical symptoms and been isolated at hospital. The isolated individuals are come from quarantined community at a rate $k_2$ and symptomatic group at a rate $\gamma_2$. The recovery rate of isolated individuals is $\sigma_4$, disease induced death rate is $\delta$ and natural death rate is $\mu$. Thus,
\begin{eqnarray}
\displaystyle{\frac{dJ}{dt}} &=& k_2 Q+\gamma_2 I-(\delta+\sigma_4+\mu)J
\end{eqnarray}
\textbf{Recovered population(R(t)):}

Quarantined, asymptomatic, symptomatic and isolated individuals recover from the disease at rates $\sigma_1$, $\sigma_2$, $\sigma_3$ and $\sigma_4$; respectively, and this population is reduced by a natural death rate $\mu$. Thus,
\begin{eqnarray}
\displaystyle{\frac{dR}{dt}} &=& \sigma_1 Q+\sigma_2 A+\sigma_3 I+\sigma_4 J-\mu R
\end{eqnarray}

 From the above considerations, the following system of ordinary differential equations governs the dynamics of the system:

\begin{eqnarray}\label{EQ:eqn 2.1}
\displaystyle{\frac{dS}{dt}} &=& \Pi-\frac{S(\beta I+r_Q \beta Q+r_A \beta A+r_J \beta J)}{N}-\mu S,\nonumber \\
\displaystyle{\frac{dE}{dt}} &=& \frac{S(\beta I+r_Q \beta Q+r_A \beta A+r_J \beta J)}{N}-(\gamma_1+k_1+\mu)E,\nonumber \\
\displaystyle{\frac{dQ}{dt}} &=& \gamma_1 E-(k_2+\sigma_1+\mu)Q,\nonumber \\
\displaystyle{\frac{dA}{dt}} &=&  p k_1 E-(\sigma_2+\mu)A, \\
\displaystyle{\frac{dI}{dt}} &=& (1-p)k_1 E-(\gamma_2+\sigma_3+\mu)I, \nonumber \\
\displaystyle{\frac{dJ}{dt}} &=& k_2 Q+\gamma_2 I-(\delta+\sigma_4+\mu)J, \nonumber\\
\displaystyle{\frac{dR}{dt}} &=& \sigma_1 Q+\sigma_2 A+\sigma_3 I+\sigma_4 J-\mu R, \nonumber
\end{eqnarray}
All the parameters and their biological interpretation are given in Table \ref{tab:mod1} respectively.\\

\begin{table}[ht]
\begin{center}
\caption{Description of parameters used in the model.}
\label{tab:mod1}
\begin{tabular}{c p{7.2cm} p{2.5cm} p{2cm}}
\hline
Parameters & Interpretation &  Value & Reference \\ \hline
$\Pi$ & Recruitment rate & 2274 &  \cite{Worldometer2020}\\
$\beta$ & Transmission rate & 0.7008 & Estimated\\
$r_Q$ & Modification factor for quarantined & 0.3  & Assumed \\
$r_A$ & Modification factor for asymptomatic & 0.45 & Assumed \\
$r_J$ & Modification factor for isolated & 0.6 &  Assumed\\
$\gamma_1$ &Rate at which the exposed individuals are diminished by quarantine & 0.0668 & Estimated\\
$\gamma_2$ &Rate at which the symptomatic individuals are diminished by isolation & 0.1059 & Estimated\\
$k_1$ &  Rate at which exposed become infected & 1/7 & \cite{Who2019} \\
$k_2$ &  Rate at which quarantined individuals are isolated & 0.0632 & Estimated\\
$p$ & Proportion of asymptomatic individuals & 0.13166 & \cite{tang2020estimation}\\
$\sigma_1$ & Recovery rate from quarantined individuals & 0.2158 & Estimated\\
$\sigma_2$ & Recovery rate from asymptomatic individuals & 0.03 & Estimated\\
$\sigma_3$ & Recovery rate from symptomatic individuals & 0.46 & \cite{Who2019}\\
$\sigma_4$ & Recovery rate from isolated individuals & 0.4521 & Estimated\\
$\delta$ & Diseases induced mortality rate & 0.0015 & \cite{Worldometer2020}\\
$\mu$ & Natural death rate  & 0.3349 $\times$ $10^{-4}$ & \cite{lifexp2018}\\
\hline
\end{tabular}
 	\end{center}
\end{table}

\section{Mathematical analysis}{\label{Mathematical_analysis}}
\subsection{\textbf{Positivity and boundedness of the solution}}
This subsection is provided to prove the positivity and boundedness of solutions of the system \eqref{EQ:eqn 2.1} with initial conditions $(S(0),E(0),Q(0),A(0),I(0),J(0),R(0))^T\in \mathbb{R}_{+}^7$. We first state the following lemma.
\begin{lemma} \label{lma1} 
     Suppose $\Omega \subset \mathbb{R} \times \mathbb{C}^n$ is open, $f_i \in C(\Omega, \mathbb{R}), i=1,2,3,...,n$. If $f_i|_{x_i(t)=0,X_t \in \mathbb{C}_{+0}^n}\geq 0$, $X_t=(x_{1t},x_{2t},.....,x_{1n})^T, i=1,2,3,....,n$, then $\mathbb{C}_{+0}^n\lbrace \phi=(\phi_1,.....,\phi_n):\phi \in \mathbb{C}([-\tau,0],\mathbb{R}_{+0}^n)\rbrace$ is the invariant domain of the following equations
	\begin{align*}
	\frac{dx_i(t)}{dt}=f_i(t,X_t), t\geq \sigma, i=1,2,3,...,n.
	\end{align*}
	where $\mathbb{R}_{+0}^n=\lbrace (x_1,....x_n): x_i\geq 0, i=1,....,n \rbrace$ \cite{yang1996permanence}.
\end{lemma}

\begin{proposition}
	The system \eqref{EQ:eqn 2.1} is invariant in $\mathbb{R}_{+}^7$.
\end{proposition}
\begin{proof}
	By re-writing the system \eqref{EQ:eqn 2.1} we have
	\begin{eqnarray}
	\frac{dX}{dt} & =M(X(t)), X(0)=X_0\geq 0
	\label{EQ:eqn 2.2}
	\end{eqnarray}
	$ M(X(t))=(M_1(X),M_1(X),...,M_7(X))^T$\\
	We note that
	\begin{align*}
	\frac{dS}{dt}|_{S=0}&=\Pi \geq 0,\\ 
	\frac{dE}{dt}|_{E=0}&=\frac{S(\beta I+r_Q \beta Q+r_A \beta A+r_J \beta J)}{S+Q+A+I+J+R}\geq 0,\\ \frac{dQ}{dt}|_{Q=0}&=\gamma_1 E \geq 0,\\ 
	\frac{dA}{dt}|_{A=0}&=p k_1 E\geq 0,\\ 
	\frac{dI}{dt}|_{I=0}&=(1-p)k_1 E\geq 0,\\ \frac{dJ}{dt}|_{J=0}&=k_2 Q+\gamma_2 I\geq 0,\\
	\frac{dR}{dt}|_{R=0}&=\sigma_1 Q+\sigma_2 A+\sigma_3 I +\sigma_4 J\geq 0.
	\end{align*}
	Then it follows from the Lemma \ref{lma1} that $\mathbb{R}_{+}^7$ is an invariant set.
\end{proof}
\begin{proposition}
	The system \eqref{EQ:eqn 2.1} is bounded in the region\\ $\Omega=\lbrace(S,E,Q,A,I,J,R)\in \mathbb{R}_+^{7}|S+E+Q+A+I+J+R\leq \frac{\Pi}{\mu}\rbrace$
\end{proposition}
\begin{proof}
	We observed from the system that
	\begin{align*}
	&\frac{dN}{dt}=\Pi-\mu N-\delta J\leq \Pi-\mu N\\
	& \Longrightarrow \lim\limits_{t\rightarrow \infty}sup N(t)\leq \frac{\Pi}{\mu}
	\end{align*}
	Hence the system \eqref{EQ:eqn 2.1} is bounded.
\end{proof}
\subsection{\textbf{Diseases-free equilibrium and control reproduction number}}
The diseases-free equilibrium can be obtained for the system \eqref{EQ:eqn 2.1} by putting $E=0, Q=0, A=0, I=0, J=0$, which is denoted by $P_1^{0}=(S^0,0,0,0,0,0,R^0),$ where\\
\begin{align*}
S^0=\frac{\Pi}{\mu}, R^0=0.
\end{align*}
The control reproduction number, a central concept in the study of the spread of communicable diseases, is e the number of secondary infections caused by a single infective in a population consisting essentially only of susceptibles with the control measures in place (quarantined and isolated class) \cite{van2008further}. This dimensionless number is calculated at the DFE by next generation operator method \cite{van2002reproduction, diekmann2000mathematical} and it is denoted by $R_c$.

 For this, we assemble the compartments which are infected from the system \eqref{EQ:eqn 2.1} and decomposing the right hand side as $\mathcal{F}-\mathcal{V}$, where $\mathcal{F}$ is the transmission part, expressing the the production of new infection, and the transition part is $\mathcal{V}$, which describe the change in state.

\begin{align*}
\mathcal{F}&=\begin{pmatrix}
\frac{S(\beta I+r_Q \beta Q+r_A \beta A+r_J \beta J)}{N}\\
0\\
0\\
0\\
0
\end{pmatrix},
\mathcal{V}=\begin{pmatrix}
(\gamma_1+k_1+\mu)E\\
-\gamma_1 E+(k_2+\sigma_1+\mu)Q\\
-p k_1 E+(\sigma_2+\mu)A\\
-(1-p)k_1 E+(\gamma_2+\sigma_3+\mu)I\\
 -k_2 Q-\gamma_2 I+(\delta+\sigma_4+\mu)J
\end{pmatrix}
\end{align*}

Now we calculate the jacobian of $\mathcal{F}$ and $\mathcal{V}$ at DFE $P_1^{0}$

\begin{align*}
F=\frac{\partial \mathcal{F}}{\partial X}=\begin{pmatrix}
0 & r_Q\beta & r_A \beta & \beta & r_J \beta \\
0&0 & 0 &0 &0  \\
0 &0 & 0 &0  &0\\
0 &0&0 &0 &0 \\
0 &0 &0 &0 &0\\
\end{pmatrix},
\end{align*}
\begin{align*}
V=\frac{\partial \mathcal{V}}{\partial X}=\begin{pmatrix}
\gamma_1+k_1+\mu &0 & 0 &0&0 \\
-\gamma_1 &k_2+\sigma_1+\mu & 0 &0 &0 \\
-p k_1 &0 & \sigma_2+\mu &0&0  \\
-(1-p)k_1 &0 &0 &\gamma_2+\sigma_3+\mu &0 \\
0 &-k_2 &0  &-\gamma_2 &\delta+\sigma_4+\mu\\
\end{pmatrix}.
\end{align*}
Following \cite{heffernan2005perspectives}, $R_c=\rho(FV^{-1})$, where $\rho$ is the spectral radius of the next-generation matrix ($FV^{-1}$). Thus, from the model \eqref{EQ:eqn 2.1}, we have the
following expression for $R_c$:
\begin{align}
R_c &=\frac{r_Q \beta \gamma_1}{(\gamma_1+k_1+\mu)(k_2+\sigma_1+\mu)}+\frac{r_A \beta p k_1}{(\gamma_1+k_1+\mu)(\sigma_2+\mu)}\\\nonumber
&+\frac{\beta k_1(1-p)}{(\gamma_1+k_1+\mu)(\gamma_2+\sigma_3+\mu)} + \frac{r_J \beta \gamma_1 k_2}{(\gamma_1+k_1+\mu)(k_2+\sigma_1+\mu)(\delta+\sigma_4+\mu)}\\\nonumber
&+\frac{r_J \beta (1-p)k_1 \gamma_2}{(\gamma_1+k_1+\mu)(\gamma_2+\sigma_3+\mu)(\delta+\sigma_4+\mu)}\\\nonumber
\end{align}
\subsection{\textbf{Stability of DFE}}
\begin{theorem}
	The diseases free equilibrium(DFE)  $P_1^{0}=(S^0,0,0,0,0,0,R^0)$ of the system \eqref{EQ:eqn 2.1} is locally asymptotically stable if $R_c<1$ and unstable if $R_c>1$.	
\end{theorem}
\begin{proof}
	We calculate the Jacobian of the system \eqref{EQ:eqn 2.1} at DFE, and is given by
\begin{align*}
J_{P_1^{0}}={\begin{pmatrix}
	-\mu &0 &-r_Q \beta &-r_A \beta & -\beta &-r_J \beta &0  \\
	0&-(\gamma_1+k_1+\mu) & r_Q\beta &r_A\beta &\beta &r_J\beta&0 \\
	0 & \gamma_1 & -(k_2+\sigma_1+\mu)&0 & 0 &0 &0  \\
	0& p k_1 &0 & -(\sigma_2+\mu) &0 &0 &0 \\
	0 &(1-p)k_1 &0 &0 &-(\gamma_2+\sigma_3+\mu)&0 &0 \\
	0 &0 &k_2 & 0 &\gamma_2 &-(\delta+\sigma_4+\mu) &0 \\
	0 &0 &\sigma_1 &\sigma_2 &\sigma_3 &\sigma_4 &-\mu \\
	\end{pmatrix}},
\end{align*}

	Let $\lambda$ be the eigenvalue of the matrix $J_{P_1^{0}}$. Then the characteristic equation is given by $det(J_{P_1^{0}}-\lambda I)=0$.\\
	$\Rightarrow$ $r_J \beta \gamma_1 k_2(\lambda+\sigma_2+\mu)(\lambda+\gamma_2+\sigma_3+\mu)+r_J \beta \gamma_2 k_1(\lambda+k_2+\sigma_1+\mu)[(1-p)(\lambda+\sigma_2+\mu)]+r_A \beta p k_1(\lambda+\gamma_2+\sigma_3+\mu)(\lambda+\delta+\sigma_4+\mu)(\lambda+k_2+\sigma_1+\mu)+\beta k_1[(1-p)(\lambda+\sigma_2+\mu)](\lambda+\delta+\sigma_4+\mu)(\lambda+k_2+\sigma_1+\mu)-(\lambda+\gamma_1+k_1+\mu)(\lambda+\sigma_2+\mu)(\lambda+\gamma_2+\sigma_3+\mu)(\lambda+\delta+\sigma_4+\mu)(\lambda+k_2+\sigma_1+\mu)=0$.\\
	Which can be written as\\
	
	\small{
	\begin{multline*}
	 \frac{r_Q \beta \gamma_1}{(\lambda+\gamma_1+k_1+\mu)(\lambda+k_2+\sigma_1+\mu)}+\frac{r_A \beta p k_1}{(\lambda+\gamma_1+k_1+\mu)(\lambda+\sigma_2+\mu)}+
	 \frac{\beta k_1(1-p)}{(\lambda+\gamma_1+k_1+\mu)(\lambda+\gamma_2+\sigma_3+\mu)}\\
	 + \frac{r_J \beta[\gamma_1 k_2(\lambda+\sigma_2+\mu)(\lambda+\gamma_2+\sigma_3+\mu)+(1-p)k_1 \gamma_2(\lambda+k_2+\sigma_1+\mu)(\lambda+\sigma_2+\mu)]}{(\lambda+\gamma_1+k_1+\mu)(\lambda+k_2+\sigma_1+\mu)(\lambda+\sigma_2+\mu)(\lambda+\gamma_2+\sigma_3+\mu)(\lambda+\delta+\sigma_4+\mu)}=1.
     \end{multline*}}

	Denote 
	\begin{align*}
	G_1(\lambda) &=\frac{r_Q \beta \gamma_1}{(\lambda+\gamma_1+k_1+\mu)(\lambda+k_2+\sigma_1+\mu)}+\frac{r_A \beta p k_1}{(\lambda+\gamma_1+k_1+\mu)(\lambda+\sigma_2+\mu)}\\\nonumber
	&+\frac{\beta k_1 (1-p)}{(\lambda+\gamma_1+k_1+\mu)(\lambda+\gamma_2+\sigma_3+\mu)}\\\nonumber
	&+ \frac{r_J \beta \gamma_1 k_2}{(\lambda+\gamma_1+k_1+\mu)(\lambda+k_2+\sigma_1+\mu)(\lambda+\delta+\sigma_4+\mu)}\\\nonumber
	&+ \frac{r_J \beta (1-p)k_1 \gamma_2}{(\lambda+\gamma_1+k_1+\mu)(\lambda+\gamma_2+\sigma_3+\mu)(\lambda+\delta+\sigma_4+\mu)}.
	\end{align*}
	We rewrite $G_1(\lambda)$ as $G_1(\lambda)=G_{11}(\lambda)+G_{12}(\lambda)+G_{13}(\lambda)+G_{14}(\lambda)+G_{15}(\lambda)$\\
	Now if $Re(\lambda)\geq 0$, $\lambda=x+iy$, then
	\begin{align*}
	|G_{11}(\lambda)|&\leq \frac{r_Q \beta \gamma_1}{|\lambda+\gamma_1+k_1+\mu||\lambda+k_2+\sigma_1+\mu|}\leq G_{11}(x)\leq G_{11}(0)\\
	|G_{12}(\lambda)|&\leq \frac{r_A \beta p k_1}{|\lambda+\gamma_1+k_1+\mu||\lambda+\sigma_2+\mu|}\leq G_{12}(x)\leq G_{12}(0)\\
	|G_{13}(\lambda)|&\leq \frac{\beta k_1 (1-p)}{|\lambda+\gamma_1+k_1+\mu||\lambda+\gamma_2+\sigma_3+\mu|}\leq G_{13}(x)\leq G_{13}(0)\\
	|G_{14}(\lambda)|&\leq \frac{r_J \beta \gamma_1 k_2}{|\lambda+\gamma_1+k_1+\mu||\lambda+k_2+\sigma_1+\mu||\lambda+\delta+\sigma_4+\mu|}\leq G_{14}(x)\leq G_{14}(0)\\
	|G_{15}(\lambda)|&\leq  \frac{r_J \beta (1-p)k_1 \gamma_2}{|\lambda+\gamma_1+k_1+\mu||\lambda+\gamma_2+\sigma_3+\mu||\lambda+\delta+\sigma_4+\mu|}\leq G_{15}(x)\leq G_{15}(0)
	\end{align*}
	
	Then $G_{11}(0)+G_{12}(0)+G_{13}(0)+G_{14}(0)+G_{15}(0)=G_1(0)=R_c<1$, which implies $|G_1(\lambda)|\leq 1$.\\
	Thus for $R_c<1$, all the eigenvalues of the characteristics equation $G_1(\lambda)=1$ has negative real parts.
	
	Therefore if $R_c<1$, all eigenvalues are negative and hence DFE $P_1^{0}$ is locally asymptotically stable.
	
	Now if we consider $R_c>1$ i.e $G_1(0)>1$, then 
	\begin{align*}
	\lim\limits_{\lambda\rightarrow \infty} G_1(\lambda)=0.
	\end{align*}
	Then there exist $\lambda_1^{*}>0$ such that $G_1(\lambda_1^{*})=1$.
	
	That means there exist positive eigenvalue $\lambda_1^{*}>0$ of the Jacobian matrix.
	
	Hence DFE $P_1^{0}$ is unstable whenever $R_c>1$.
\end{proof}

\begin{theorem}
	The diseases free equilibrium (DFE) $P_1^{0}=(S^0,0,0,0,0,0,R^0)$  is globally asymptotically stable (GAS) for the system \eqref{EQ:eqn 2.1} if $R_c<1$ and unstable if $R_c>1$.
\end{theorem}
\begin{proof}
	We rewrite the system \eqref{EQ:eqn 2.1} as
	\begin{align*}
	\frac{dX}{dt}&=F(X,V)\\
	\frac{dV}{dt}&=G(X,V), G(X,0)=0
	\end{align*}	
	where $X=(S, R)\in R_2$ (the number of uninfected individuals compartments),
	$V=(E, Q, A, I, J)\in R_5 $ (the number of infected individuals compartments), and $P_1^{0}=(\frac{\Pi}{\mu},0,0,0,0,0,0)$ is the DFE of the system \eqref{EQ:eqn 2.1}. The global stability of the DFE is guaranteed if the following two conditions are satisfied:
	
	\begin{enumerate}
		\item For $\frac{dX}{dt}=F(X,0)$, $X^*$ is globally asymptotically stable,
		\item $G(X,V) = BV-\widehat{G}(X,V),$ $\widehat{G}(X,V)\geq 0$ for $(X,V)\in \Omega$,
	\end{enumerate}
	where $B=D_VG(X^*,0)$ is a Metzler matrix and $\Omega$ is the positively invariant set with respect to the model \eqref{EQ:eqn 2.1}. Following Castillo-Chavez et al \cite{castillo2002computation}, we check for aforementioned conditions.\\
	For system \eqref{EQ:eqn 2.1},
	\begin{equation*}
	F(X,0)=\begin{pmatrix}
	\Pi -\mu S\\
	0
	\end{pmatrix},
	\end{equation*}
	
	\begin{equation*}
	B=\begin{pmatrix}
	-(\gamma_1+k_1+\mu) & r_Q \beta & r_A \beta & \beta & r_J \beta \\
	\gamma_1 & -(k_2+\sigma_1+\mu) & 0 & 0 &0\\
	p k_1 & 0 & -(\sigma_2+\mu) & 0& 0\\
	(1-p)k_1 & 0 & 0 & -(\gamma_2+\sigma_3+\mu)& 0\\
	0 & k_2 & 0 & \gamma_2 & -(\delta+\sigma_4+\mu)
	\end{pmatrix}
	\end{equation*}
	and 
	\begin{align*}
	\widehat{G}(X,V)=\begin{pmatrix}
	r_Q \beta Q (1-\frac{S}{N})+r_A \beta A (1-\frac{S}{N})+ \beta I(1-\frac{S}{N})+r_J \beta J (1-\frac{S}{N})\\
	0\\
	0\\
	0\\
	0
	\end{pmatrix}.
	\end{align*}
	
	Clearly, $\widehat{G}(X,V)\geq 0$ whenever the state variables are inside $\Omega$. Also it is clear that $X^*=(\frac{\Pi}{\mu},0)$ is a globally asymptotically stable equilibrium of the system $\frac{dX}{dt}=F(X,0)$. Hence, the theorem follows.
\end{proof}
\subsection{\textbf{Existence and local stability of endemic equilibrium}}

In this section, the existence of the endemic equilibrium of the model \eqref{EQ:eqn 2.1} is established. Let us denote
\begin{align*}
    m_1 & =\gamma_1+k_1+\mu, m_2=k_2+\sigma_1+\mu, m_3=\sigma_2+\mu,\\
    m_4&=\gamma_2+\sigma_3+\mu, m_5=\delta+\sigma_4+\mu.
\end{align*}
Let $P^*=(S^*, E^*, Q^*, A^*, I^*, J^*, R^*)$ represents any arbitrary endemic equilibrium point (EEP) of the model \eqref{EQ:eqn 2.1}. Further, define
\begin{align}\label{EQ:eqn 3.30}
    \eta^*=\frac{\beta(I^*+r_Q Q^*+r_A A^*+r_J J^*)}{N^*}
\end{align}
It follows, by solving the equations in \eqref{EQ:eqn 2.1} at steady-state, that
\begin{align}\label{EQ:eqn 3.40}
    S^*&=\frac{\Pi}{\eta^*+\mu}, E^*=\frac{\eta^*S^*}{m_1}, Q^*=\frac{\gamma_1 \eta^* S^*}{m_1 m_2}, A^*=\frac{p k_1 \eta^* S^*}{m_1 m_3},\\\nonumber
    I^*& =\frac{(1-p)k_1 \eta^* S^*}{m_1 m_4}, J^*=\frac{\eta^* S^*(k_2 \gamma_1 m_4+(1-p)k_1 \gamma_2 m_2)}{m_1 m_2 m_4 m_5}\\
    R^* &=\frac{\eta^* S^*[\sigma_1 \gamma_1 m_3 m_4 m_5+p k_1 \sigma_2 m_2 m_4 m_5+(1-p)k_1 \sigma_3 m_2 m_3 m_5 +m_3 \sigma_4(k_2 \gamma_1 m_4+(1-p)k_1 \gamma_2 m_2)]}{\mu m_1 m_2 m_3 m_4 m_5}\nonumber
\end{align}
Substituting the expression in \eqref{EQ:eqn 3.40} into \eqref{EQ:eqn 3.30} shows that the non-zero equilibrium of the model \eqref{EQ:eqn 2.1} satisfy the following linear equation, in terms of $\eta^*$:
\begin{align}
    A \eta^* +B=0
\end{align}
where 
\begin{align*}
    A&= \mu[m_2 m_3 m_4 m_5+\gamma_1 m_3 m_4 m_5+p k_1 m_2 m_4 m_5+(1-p)k_1 m_2 m_3 m_5+k_2 \gamma_1 m_3 m_4\\ & +(1-p)k_1 \gamma_2 m_2 m_3]+\sigma_1 \gamma_1 m_3 m_4 m_5
     +\sigma_2 p k_1 m_2 m_4 m_5+(1-p)k_1 \sigma_3 m_2 m_3 m_5\\ &+\sigma_4 k_2 \gamma_1 m_3 m_4+(1-p)\sigma_4 \gamma_2 k_1 m_2 m_3\\
    B&=\mu m_1 m_2 m_3 m_4 m_5(1-R_c)
\end{align*}
Since $A>0$, $\mu>0$, $m_1>0$, $m_2>0$, $m_3>0$, $m_4>0$ and $m_5>0$, it is clear that the model \eqref{EQ:eqn 2.1} has a unique endemic equilibrium point (EEP) whenever $R_c>1$ and no positive endemic equilibrium point whenever $R_c<1$. This rules out the possibility of the existence of equilibrium other than DFE whenever $R_c<1$. Furthermore, it can be shown that, the DFE $P_1^{0}$ of the model \eqref{EQ:eqn 2.1} is globally asymptotically stable (GAS) whenever $R_c<1$.

From the above discussion we have concluded that
\begin{theorem}
     The model \eqref{EQ:eqn 2.1} has a unique endemic (positive) equilibrium, given by $P^*$, whenever $R_c>1$ and has no endemic equilibrium for $R_c\leq 1$.
\end{theorem}

Now we will prove the local stability of endemic equilibrium.
\begin{theorem}
   The endemic equilibrium $P^*$ is locally asymptotically stable if $R_C>1$.
\end{theorem}
\begin{proof}
Let $x = (x_1 , x_2 , x_3 , x_4 , x_5 , x_6, x_7)^T = ( S , E , Q, A, I , J, R )^T$. Thus, the model \eqref{EQ:eqn 2.1} can be re-written in the form $\frac{dx}{dt}=f(x)$, with $f ( x ) = ( f_1 ( x ) ,....., f_7 ( x ))$, as follows:

\begin{eqnarray}\label{EQ:eqn 3.81}
\displaystyle{\frac{dx_1}{dt}} &=& \Pi-\frac{x_1(\beta x_5+r_Q \beta x_3+r_A \beta x_4+r_J \beta x_6)}{x_1+x_2+x_3+x_4+x_5+x_6+x_7}-\mu x_1 ,\nonumber \\
\displaystyle{\frac{dx_2}{dt}} &=& \frac{x_1(\beta x_5+r_Q \beta x_3+r_A \beta x_4+r_J \beta x_6)}{x_1+x_2+x_3+x_4+x_5+x_6+x_7}-(\gamma_1+k_1+\mu)x_2,\nonumber \\
\displaystyle{\frac{dx_3}{dt}} &=& \gamma_1 x_2-(k_2+\sigma_1+\mu)x_3,\nonumber \\
\displaystyle{\frac{dx_4}{dt}} &=&  p k_1 x_2-(\sigma_2+\mu)x_4, \\
\displaystyle{\frac{dx_5}{dt}} &=& (1-p)k_1 x_2-(\gamma_2+\sigma_3+\mu)x_5, \nonumber \\
\displaystyle{\frac{dx_6}{dt}} &=& k_2 x_3+\gamma_2 x_5-(\delta+\sigma_4+\mu)x_6, \nonumber\\
\displaystyle{\frac{dx_7}{dt}} &=& \sigma_1 x_3+\sigma_2 x_4+\sigma_3 x_5+\sigma_4 x_6-\mu x_7, \nonumber
\end{eqnarray}

    The Jacobian matrix of the system \eqref{EQ:eqn 3.81} $J_{P_1^{0}}$ at DFE is given by
\begin{align*}
J_{P_1^{0}}={\begin{pmatrix}
	-\mu &0 &-r_Q \beta &-r_A \beta & -\beta &-r_J \beta &0  \\
	0&-(\gamma_1+k_1+\mu) & r_Q\beta &r_A\beta &\beta &r_J\beta&0 \\
	0 & \gamma_1 & -(k_2+\sigma_1+\mu)&0 & 0 &0 &0  \\
	0& p k_1 &0 & -(\sigma_2+\mu) &0 &0 &0 \\
	0 &(1-p)k_1 &0 &0 &-(\gamma_2+\sigma_3+\mu)&0 &0 \\
	0 &0 &k_2 & 0 &\gamma_2 &-(\delta+\sigma_4+\mu) &0 \\
	0 &0 &\sigma_1 &\sigma_2 &\sigma_3 &\sigma_4 &-\mu \\
	\end{pmatrix}},
\end{align*}

Here, we use the central manifold theory method to determine the local stability of the endemic equilibrium by taking $\beta$ as bifurcation parameter \cite{castillo2004dynamical}.  Select $\beta$ as the bifurcation parameter and gives critical value of $\beta$ at $R_C=1$ is given as

\begin{equation*}
\beta^*=\frac{(\gamma_1+k_1+\mu)(k_2+\sigma_1+\mu)(\sigma_2+\mu)(\gamma_2+\sigma_3+\mu)(\delta+\sigma_4+\mu)}{[r_Q \gamma_1(\sigma_2+\mu)(\gamma_2+\sigma_3+\mu)(\delta+\sigma_4+\mu)+r_A p k_1 (k_2+\sigma_1+\mu)(\gamma_2+\sigma_3+\mu)(\delta+\sigma_4+\mu) + Z]}    
\end{equation*}
where,
$Z=k_1(1-p)(k_2+\sigma_1+\mu)(\sigma_2+\mu)(\delta+\sigma_4+\mu)+r_J \gamma_1 k_2(\sigma_2+\mu)(\gamma_2+\sigma_3+\mu)
    +r_J(1-p)k_1 \gamma_2(k_2+\sigma_1+\mu)(\sigma_2+\mu)$

The Jacobian of \eqref{EQ:eqn 2.1} at $\beta=\beta^*$, denoted by $J_{P_1^{0}}|_{\beta=\beta^*}$ has a right eigenvector (corresponding to the zero eigenvalue) given by
$w=(w_1, w_2, w_3, w_4, w_5, w_6, w_7)^T$ , where
\begin{align*}
    & w_1 =-\frac{\gamma_1+k_1+\mu}{\mu}w_2, w_2=w_2>0, w_3=\frac{\gamma_1}{k_2+\sigma_1+\mu}w_2, w_4=\frac{p k_1}{\sigma_2+\mu}w_2,\\
    & w_5 =\frac{(1-p)k_1}{\gamma_2+\sigma_3+\mu}w_2, w_6=\frac{k_2 \gamma_1}{(\delta+\sigma_4+\mu)(k_2+\sigma_1+\mu)}w_2+\frac{\gamma_2(1-p)k_1}{(\delta+\sigma_4+\mu)(\gamma_2+\sigma_3+\mu)}w_2\\
    & w_7=\frac{1}{\mu}\Big[\frac{\sigma_1 \gamma_1}{k_2+\sigma_1+\mu}w_2+\frac{\sigma_2 p k_1}{\sigma_2+\mu}w_2+\frac{\sigma_3(1-p)k_1}{\gamma_2+\sigma_3+\mu]w_2}+\frac{\sigma_4 k_2 \gamma_1}{(\delta+\sigma+\mu)(k_2+\sigma_1+\mu)}w_2\\
    &+\frac{\sigma_4 \gamma_2(1-p)k_1}{(\delta+\sigma+\mu)(\gamma_2+\sigma_3+\mu)}w_2\Big].
    \end{align*}
    Similarly, from  $J_{P_1^{0}}|_{\beta=\beta^*}$, we obtain a left eigenvector $v=(v_1, v_2, v_3, v_4, v_5, v_6, v_7)$ (corresponding to the zero eigenvalue), where
    \begin{align*}
        & v_1=0, v_2=v_2>0, v_3=\frac{r_Q \beta^*}{k_2+\sigma_1+\mu}v_2+\frac{k_2 r_J \beta^*}{(k_2+\sigma_1+\mu)(\delta+\sigma_4+\mu)}v_2, v_4=\frac{r_A \beta^*}{\sigma_2+\mu}v_2,\\
        & v_5=\frac{\beta^*}{\gamma_2+\sigma_3+\mu}v_2+\frac{\gamma_2 r_J \beta^*}{(\gamma_2+\sigma_3+\mu)(\delta+\sigma_4+\mu)}v_2, v_6=\frac{r_J \beta^*}{\delta+\sigma_4+\mu}v_2, v_7=0.
    \end{align*}
    
    We calculate the following second order partial derivatives of $f_i$ at the disease-free equilibrium $P_1^0$ to show the stability of the endemic equilibrium and obtain
\begin{align}\nonumber
\frac{\partial^2 f_2}{\partial x_3 \partial x_2}&=-\frac{\beta r_Q \mu}{\Pi}, \frac{\partial^2 f_2}{\partial x_4 \partial x_2}=-\frac{\beta r_A \mu}{\Pi},
 \frac{\partial^2 f_2}{\partial x_5 \partial x_2}=-\frac{\beta \mu}{\Pi},
 \frac{\partial^2 f_2}{\partial x_6 \partial x_2}=-\frac{\beta r_J \mu}{\Pi},\\\nonumber
\frac{\partial^2 f_2}{\partial x_2 \partial x_3}&=-\frac{\beta r_Q \mu}{\Pi},
\frac{\partial^2 f_2}{\partial x_3 \partial x_3}=-\frac{2\beta r_Q \mu}{\pi},
\frac{\partial^2 f_2}{\partial x_4 \partial x_3}=-\frac{\beta r_Q \mu}{\Pi}-\frac{\beta r_A \mu}{\Pi},\\\nonumber
\frac{\partial^2 f_2}{\partial x_5 \partial x_3}&=-\frac{\beta r_Q \mu}{\Pi}-\frac{\beta \mu}{\pi},
\frac{\partial^2 f_2}{\partial x_6 \partial x_3}=-\frac{\beta r_Q \mu}{\pi}-\frac{\beta r_J \mu}{\Pi},
\frac{\partial^2 f_2}{\partial x_7 \partial x_3}=-\frac{\beta r_Q \mu}{\Pi}, \frac{\partial^2 f_2}{\partial x_2 \partial x_4}=-\frac{\beta r_A \mu}{\Pi},\\\nonumber
\frac{\partial^2 f_2}{\partial x_3 \partial x_4}&=-\frac{\beta r_A \mu}{\Pi}-\frac{\beta r_Q \mu}{\Pi}, \frac{\partial^2 f_2}{\partial x_4 \partial x_4}=-\frac{2 \beta r_A \mu}{\Pi},
\frac{\partial^2 f_2}{\partial x_5 \partial x_4}=-\frac{\beta r_A \mu}{\Pi}-\frac{\beta \mu}{\Pi},\\\nonumber
\frac{\partial^2 f_2}{\partial x_6 \partial x_4}&=-\frac{\beta r_A \mu}{\Pi}-\frac{\beta r_J \mu}{\Pi},
\frac{\partial^2 f_2}{\partial x_7 \partial x_4}=-\frac{\beta r_A \mu}{\Pi}, \frac{\partial^2 f_2}{\partial x_2 \partial x_5}=-\frac{\beta \mu}{\Pi},\\\nonumber
\frac{\partial^2 f_2}{\partial x_3 \partial x_5}&=-\frac{\beta \mu}{\Pi}-\frac{\beta r_Q \mu}{\Pi}, \frac{\partial^2 f_2}{\partial x_4 \partial x_5}=-\frac{\beta \mu}{\Pi}-\frac{\beta r_A \mu}{\Pi}, \frac{\partial^2 f_2}{\partial x_5 \partial x_5}=-\frac{2 \beta \mu}{\Pi},\\\nonumber
\frac{\partial^2 f_2}{\partial x_6 \partial x_5}&=-\frac{\beta \mu}{\pi}-\frac{\beta r_J \mu}{\Pi},
\frac{\partial^2 f_2}{\partial x_7 \partial x_5}=-\frac{\beta \mu}{\Pi}, \frac{\partial^2 f_2}{\partial x_2 \partial x_6}=-\frac{\beta r_J \mu}{\Pi}, \frac{\partial^2 f_2}{\partial x_3 \partial x_6}=-\frac{\beta r_J \mu}{\Pi}-\frac{\beta r_Q \mu}{\Pi},\\\nonumber
\frac{\partial^2 f_2}{\partial x_4 \partial x_6}&=-\frac{\beta r_J \mu}{\Pi}-\frac{\beta r_A \mu}{\Pi},
\frac{\partial^2 f_2}{\partial x_5 \partial x_6}=-\frac{\beta r_J \mu}{\Pi}-\frac{\beta \mu}{\Pi}, \frac{\partial^2 f_2}{\partial x_6 \partial x_6}=-\frac{2 \beta r_J \mu}{\Pi}, \frac{\partial^2 f_2}{\partial x_7 \partial x_6}=-\frac{\beta r_J \mu}{\Pi},\\\nonumber
\frac{\partial^2 f_2}{\partial x_3 \partial x_7}&=-\frac{\beta r_Q \mu}{\Pi}, \frac{\partial^2 f_2}{\partial x_4 \partial x_7}=\frac{\beta r_A \mu}{\Pi},
\frac{\partial^2 f_2}{\partial x_5 \partial x_7}=-\frac{\beta \mu}{\Pi}, \frac{\partial^2 f_2}{\partial x_6 \partial x_7}=-\frac{\beta r_J \mu}{\Pi}
\end{align}
    Now we calculate the coefficients $a$ and $b$ defined in Theorem 4.1 \cite{castillo2004dynamical} of Castillo–Chavez and Song as follow
    \begin{align*}
        a= \sum_{k,i,j=1}^{7} v_k w_i w_j\frac{\partial^2 f_k(0, 0)}{\partial x_i \partial x_j}
    \end{align*}
    and
    \begin{align*}
        b= \sum_{k,i=1}^{7} v_k w_i\frac{\partial^2 f_k(0, 0)}{\partial x_i \partial \beta}
    \end{align*}
    Replacing the values of all the second-order derivatives measured at DFE and $\beta=\beta^*$, we get
    \begin{align*}
        a&=-\frac{2 \beta^* \mu v_2}{\Pi}(r_Q w_3+r_A w_4+w_5+r_J w_6)(w_2+w_3+w_4+w_5+w_6+w_7)<0
    \end{align*}
    and 
    \begin{align*}
        b&= v_2(r_Q w_3+r_A w_4+w_5+r_J w_6)>0
    \end{align*}
    Since $a<0$ and $b>0$ at $\beta=\beta^*$, therefore using the Remark 1 of the Theorem 4.1 stated in \cite{castillo2004dynamical}, a transcritical bifurcation occurs at $R_C=1$ and the unique endemic equilibrium is locally asymptotically stable for $R_C>1$.
\end{proof}

The transcritical bifurcation diagram is depicted in Fig. \ref{Fig:threshold_r0}.

\begin{figure}[t]
	\includegraphics[width=1.0\textwidth]{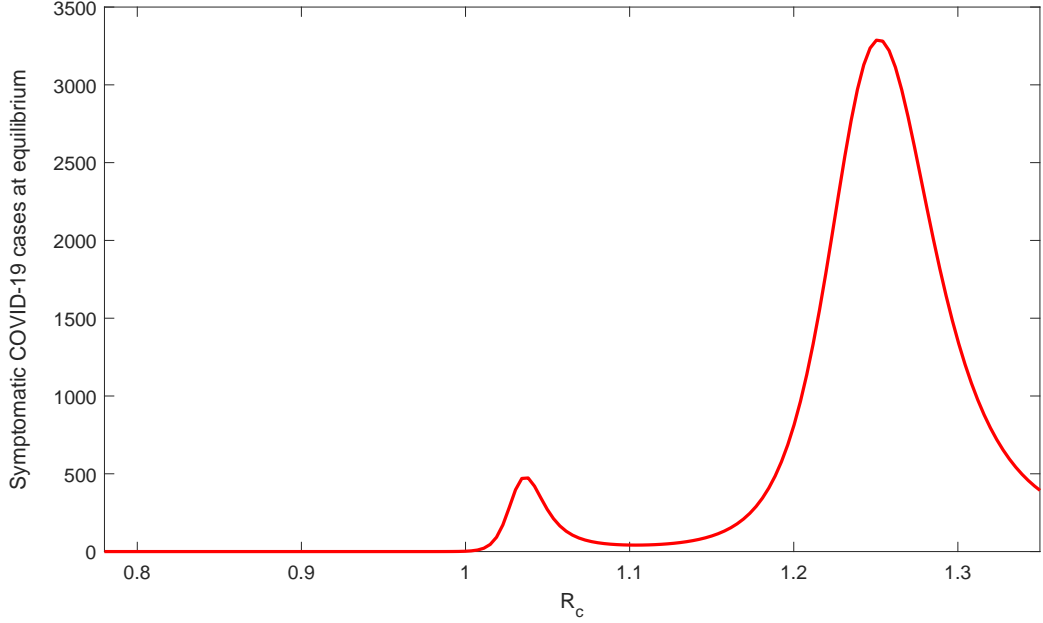}
	\caption{Forward bifurcation diagram with respect to $R_c$. All the fixed parameters are taken from Table \ref{tab:mod1} with $\gamma_1=0.0001$, $\gamma_2=0.0001$, $k_2=0.0632$, $\sigma_1=0.2158$, $\sigma_2=0.03$ $\sigma_4=0.4521$ and $0.2 < \beta < 0.35$.}
	\label{Fig:threshold_r0}
\end{figure}

\subsection{\textbf{Threshold analysis}}
In this section the impact of quarantine and isolation is measured qualitatively on the disease transmission dynamics. A threshold study of the parameters correlated with the quarantine of exposed individuals $\gamma_1 $and the isolation of the infected symptomatic individuals $\gamma_2$ is performed by measuring the partial derivatives of the control reproduction number $R_c$ with respect to these parameters. We observe that
\begin{align*}
    \frac{\partial R_c}{\partial \gamma_1}&=\frac{r_Q \beta(k_1+\mu)}{(\gamma_1+k_1+\mu)^2(k_2+\sigma_1+\mu)}
     -\frac{r_A \beta p k_1}{(\gamma_1+k_1+\mu)^2(\sigma_2+\mu)}-\frac{\beta k_1 (1-p)}{(\gamma_1+k_1+\mu)^2(\gamma_2+\sigma_3+\mu)}\\
     &+\frac{r_J \beta}{(\gamma_1+k_1+\mu)^2(\delta+\sigma_4+\mu)}\Big[\frac{k_2(k_1+\mu)}{k_2+\sigma_1+\mu}-\frac{(1-p)k_1 \gamma_2}{\gamma_2+\sigma_3+\mu}\Big]
\end{align*}
so that,
    $\frac{\partial R_c}{\partial \gamma_1}<0$ $(>0)$ iff $r_Q<r_{\gamma_1}$ ($r_Q>r_{\gamma_1}$)\\
    where
\begin{align*}
    0<r_{\gamma_1}&=\frac{k_2+\sigma_1+\mu}{k_1+\mu}\Big[\frac{r_A p k_1}{\sigma_2+\mu}+\frac{k_1(1-p)}{\gamma_2+\sigma_3+\mu}\Big]\\
    & +\frac{r_J(k_2+\sigma_1+\mu)}{(k_1+\mu)(\delta+\sigma_4+\mu)}\Big[\frac{(1-p)k_1 \gamma_2}{\gamma_2+\sigma_3+\mu}-\frac{k_2(k_1+\mu)}{k_2+\sigma_1+\mu}\Big]
\end{align*}
From the previous analysis it is obvious that if the relative infectiousness of quarantine individuals $r_Q$ will not cross the threshold value $r_{\gamma_1}$, then quarantining of exposed individuals results in reduction of the control reproduction number $R_c$ and therefore reduction of the disease burden. On the other side, if $r_Q > r_{\gamma_1}$, then the control reproduction number $R_c$ would rise due to the increase in the quarantine rate and thus the disease burden will also rise and therefore the use of quarantine in this scenario is harmful. The result is summarized in the following way:
\begin{theorem}\label{EQ:eqn 3.7}
    For the model \eqref{EQ:eqn 2.1}, the use of quarantine of the exposed individuals will have positive (negative) population-level impact if $r_Q<r_{\gamma_1}$ $(r_Q>r_{\gamma_1})$.
\end{theorem}

Similarly, measuring the partial derivatives of $R_c$ with respect to the isolation parameter $\gamma_2$is used to determine the effect of isolation of infected symptomatic individuals. Thus, we obtain
\begin{align*}
    \frac{\partial R_c}{\partial \gamma_2}&=\frac{r_J \beta(1-p)k_1}{(\gamma_1+k_1+\mu)(\gamma_2+\sigma_3+\mu)(\delta+\sigma_4+\mu)}-\frac{r_J \beta (1-p)k_1 \gamma_2}{(\gamma_1+k_1+\mu)(\gamma_2+\sigma_3+\mu)^2(\delta+\sigma_4+\mu)}\\
    & -\frac{\beta k_1 (1-p)}{(\gamma_1+k_1+\mu)(\gamma_2+\sigma_3+\mu)^2}
\end{align*}
Thus,
    $\frac{\partial R_c}{\partial \gamma_2}<0$ $(>0)$ iff $r_J<r_{\gamma_2}$ ($r_J>r_{\gamma_2}$)\\
    where
    \begin{align*}
        0<r_{\gamma_2}=\frac{\delta+\sigma_4+\mu}{\sigma_3+\mu}
    \end{align*}
The use of isolation of infected symptomatic individuals will also be effective in controlling the disease in the population if the relative infectiousness of the isolated individuals $r_J$ does not cross the threshold $r_{\gamma_2}$. The result is summarized below:
\begin{theorem}\label{EQ:eqn 3.8}
    For the model \eqref{EQ:eqn 2.1}, the use of isolation of infected symptomatic individuals will have positive (negative) population-level impact if $r_J<r_{\gamma_2}$ $(r_J>r_{\gamma_2})$.
\end{theorem}

The control reproduction number $R_c$ is a decreasing (non-decreasing) function of the quarantine and isolation parameters $\gamma_1$ and $\gamma_2$ if the conditions $r_Q<r_{\gamma_1}$ and $r_J<r_{\gamma_2}$ are respectively satisfied. See figure \ref{fig:threshold_gamma1_gamma2}(a) and \ref{fig:threshold_gamma1_gamma2}(b) obtained from model simulation in which the results correspond to the theoretical findings discussed.

\subsection{\textbf{Model without control and basic reproduction number}}
We consider the system in this section when there is no control mechanism, that is, in the absence of quarantined and isolated classes. Setting $\gamma_1=\gamma_2=0$ in the model \eqref{EQ:eqn 2.1} give the following reduce model
\begin{eqnarray}\label{EQ:eqn 3.3}
\displaystyle{\frac{dS}{dt}} &=& \Pi-\frac{S(\beta I+r_A \beta A)}{\hat{N}}-\mu S,\nonumber \\
\displaystyle{\frac{dE}{dt}} &=& \frac{S(\beta I+r_A \beta A)}{\hat{N}}-(k_1+\mu)E,\nonumber \\
\displaystyle{\frac{dA}{dt}} &=&  p k_1 E-(\sigma_2+\mu)A, \\
\displaystyle{\frac{dI}{dt}} &=& (1-p)k_1 E-(\sigma_3+\mu)I, \nonumber \\
\displaystyle{\frac{dR}{dt}} &=& \sigma_2 A+\sigma_3 I-\mu R, \nonumber
\end{eqnarray}
Where $\hat{N}= S+E+A+I+R$.
The diseases-free equilibrium can be obtained for the system \eqref{EQ:eqn 3.3} by putting $E=0, A=0, I=0$, which is denoted by $P_2^{0}=(S^0,0,0,0,R^0),$ where\\
\begin{align*}
S^0=\frac{\Pi}{\mu}, R^0=0.
\end{align*}
We will follow the convention that the basic reproduction number is defined in the absence of control measure, denoted by $R_0$ whereas we calculate the control reproduction number when the control measure are in the place. The basic reproduction number $R_0$ is defined as the expected
number of secondary infections produced by a single infected individual in a fully susceptible population during his infectious period \cite{anderson1991may, diekmann2000mathematical,hethcote2000mathematics}. We calculate $R_0$ in the same way as we calculate $R_c$ by using next generation operator method \cite{van2002reproduction}. Now we calculate the jacobian of $\mathcal{F}$ and $\mathcal{V}$ at DFE $P_2^{0}$

\begin{align*}
F=\frac{\partial \mathcal{F}}{\partial X}=\begin{pmatrix}
0 & r_A \beta & \beta \\
0&0 & 0   \\
0 &0 & 0 \\
\end{pmatrix},
V=\frac{\partial \mathcal{V}}{\partial X}=\begin{pmatrix}
\gamma_1+k_1+\mu &0 & 0 \\
-p k_1 & \sigma_2+\mu &0  \\
-(1-p)k_1 &0  &\gamma_2+\sigma_3+\mu  \\
\end{pmatrix}.
\end{align*}
Following \cite{heffernan2005perspectives}, $R_0=\rho(FV^{-1})$, where $\rho$ is the spectral radius of the next-generation matrix ($FV^{-1}$). Thus, from the model \eqref{EQ:eqn 3.3}, we have the
following expression for $R_0$:
\begin{align}
R_0&=\frac{r_A \beta p k_1}{(k_1+\mu)(\sigma_2+\mu)}+\frac{\beta k_1(1-p)}{(k_1+\mu)(\sigma_3+\mu)}
\end{align}

 Thus, $R_0$ is $R_c$ with $\gamma_1= \gamma_2=0$.
 \subsubsection{\textbf{Stability of DFE of the model \ref{EQ:eqn 3.3}}}
 \begin{theorem}
 	The diseases free equilibrium (DFE)  $P_2^{0}=(S^0,0,0,0,R^0)$ of the system \eqref{EQ:eqn 3.3} is locally asymptotically stable if $R_0<1$ and unstable if $R_0>1$.	
 \end{theorem}
 \begin{proof}
 	We calculate the Jacobian of the system \eqref{EQ:eqn 3.3} at DFE $P_2^{0}$, is given by
 	\begin{align*}
 	J_{P_2^{0}}={\begin{pmatrix}
 		-\mu &0  &-r_A \beta & -\beta  &0  \\
 		0&-(k_1+\mu) &r_A\beta &\beta &0 \\
 		0& p k_1  & -(\sigma_2+\mu) &0  &0 \\
 		0 &(1-p)k_1 &0 &-(\sigma_3+\mu) &0 \\
 		0 &0  &\sigma_2 &\sigma_3  &-\mu \\
 		\end{pmatrix}}
 	\end{align*}
 	Let $\lambda$ be the eigenvalue of the matrix $J_{P_2^{0}}$. Then the characteristic equation is given by $det(J_{P_2^{0}}-\lambda I)=0$.
 	\begin{multline*}
 	\Rightarrow r_A \beta p k_1(\lambda+\sigma_3+\mu)+\beta k_1[(1-p)(\lambda+\sigma_2+\mu)]-(\lambda+k_1+\mu)(\lambda+\sigma_2+\mu)(\lambda+\sigma_3+\mu)=0.
 	\end{multline*}
 	which implies
 	\begin{multline*}
 	    \frac{r_A \beta p k_1}{(\lambda+k_1+\mu)(\lambda+\sigma_2+\mu)}+\frac{\beta k_1 (1-p)}{(\lambda+k_1+\mu)(\lambda+\sigma_3+\mu)}=1.\\
 	\end{multline*}
 	
 	Denote 
 	\begin{align*}
 	G_2(\lambda) &=\frac{r_A \beta p k_1}{(\lambda+k_1+\mu)(\lambda+\sigma_2+\mu)}+\frac{\beta k_1 (1-p)}{(\lambda+k_1+\mu)(\lambda+\sigma_3+\mu)}.
 	\end{align*}
 	We rewrite $G_2(\lambda)$ as $G_2(\lambda)=G_{21}(\lambda)+G_{22}(\lambda)$\\
 		Now if $Re(\lambda)\geq 0$, $\lambda=x+iy$, then
 		\begin{align*}
 		|G_{21}(\lambda)|&\leq \frac{r_A \beta p k_1}{|\lambda+k_1+\mu||\lambda+\sigma_2+\mu|} \leq G_{21}(x)\leq G_{21}(0)\\
 		|G_{22}(\lambda)|&\leq \frac{\beta k_1 (1-p)}{|\lambda+k_1+\mu||\lambda+\sigma_3+\mu|}\leq G_{22}(x)\leq G_{22}(0)
 		\end{align*}
 		
 		Then $G_{21}(0)+G_{22}(0)=G_2(0)=R_0<1$, which implies $|G_2(\lambda)|\leq 1$.\\
 		Thus for $R_0<1$, all the eigenvalues of the characteristics equation $G_2(\lambda)=1$ has negative real parts.
 		
 		Therefore if $R_0<1$, all eigenvalues are negative and hence DFE $P_2^{0}$ is locally asymptotically stable.
 		
 		Now if we consider $R_0>1$ i.e $G_2(0)>1$, then 
 		\begin{align*}
 		\lim\limits_{\lambda\rightarrow \infty} G_2(\lambda)=0.
 		\end{align*}
 		Then there exist $\lambda^*>0$ such that $G_2(\lambda^*)=1$.
 		
 		That means there exist positive eigenvalue $\lambda^*>0$ of the Jacobian matrix.
 		
 		Hence DFE $P_2^{0}$ is unstable whenever $R_0>1$.
 	\end{proof}
\begin{theorem}
	The diseases free equilibrium (DFE) $P_2^{0}=(S^0,0,0,0,R^0)$  is globally asymptotically stable for the system \eqref{EQ:eqn 3.3} if $R_0<1$ and unstable if $R_0>1$.
\end{theorem}
\begin{proof}
	We rewrite the system \eqref{EQ:eqn 3.3}as
	\begin{align*}
	\frac{dX}{dt}&=F_1(X,V)\\
	\frac{dV}{dt}&=G_1(X,V), G_1(X,0)=0
	\end{align*}	
	where $X=(S, R)\in R_2$ (the number of uninfected individuals compartments),
	$V=(E, A, I)\in R_3 $ (the number of infected individuals compartments), and $P_2^{0}=(\frac{\Pi}{\mu},0,0,0,0)$ is the DFE of the system \eqref{EQ:eqn 3.3}. The global stability of the DFE is guaranteed if the following two conditions are satisfied:
	
	\begin{enumerate}
		\item For $\frac{dX}{dt}=F_1(X,0)$, $X^*$ is globally asymptotically stable,
		\item $G_1(X,V) = BV-\widehat{G}_1(X,V),$ $\widehat{G}_1(X,V)\geq 0$ for $(X,V)\in \hat{\Omega}$,
	\end{enumerate}
	where $B=D_VG_1(X^*,0)$ is a Metzler matrix and $\hat{\Omega}$ is the positively invariant set with respect to the model \eqref{EQ:eqn 3.3}. Following Castillo-Chavez et al \cite{castillo2002computation}, we check for aforementioned conditions.\\
	For system \eqref{EQ:eqn 3.3},
	\begin{align*}
	F_1(X,0)&=\begin{pmatrix}
	\Pi -\mu S\\
	0
	\end{pmatrix},\\
	B&=\begin{pmatrix}
	-(k_1+\mu) & r_A \beta & \beta \\
	p k_1 & -(\sigma_2+\mu) & 0\\
	(1-p)k_1 & 0 & -(\sigma_3+\mu)
	\end{pmatrix}\\
	\end{align*}
	and 
	\begin{align*}
	\widehat{G}_1(X,V)=\begin{pmatrix}
	r_A \beta A (1-\frac{S}{\hat{N}})+ \beta I(1-\frac{S}{\hat{N}})\\
	0\\
	0
	\end{pmatrix}.
	\end{align*}
	
	Clearly, $\widehat{G}_1(X,V)\geq 0$ whenever the state variables are inside $\hat{\Omega}$. Also it is clear that $X^*=(\frac{\Pi}{\mu},0)$ is a globally asymptotically stable equilibrium of the system $\frac{dX}{dt}=F_1(X,0)$. Hence, the theorem follows.
\end{proof}

\section{Model Calibration and epidemic potentials}
We calibrated our model \eqref{EQ:eqn 2.1} to the daily new COVID-19 cases for the UK. Daily COVID-19 cases are collected for the period 6 March, 2020 - 30 June, 2020 \cite{Worldometer2020}. We divide the 116 data points into training period and testing periods, viz., 6 March - 15 June and 16 June - 30 June respectively. We fit the model \eqref{EQ:eqn 2.1} to daily new isolated cases of COVID-19 in the UK. Due to the highly transmissible virus, the notified cases are immediately isolated, and therefore it is convenient to fit the isolated cases to reported data. Also we fit the model \eqref{EQ:eqn 2.1} to cumulative isolated cases of COVID-19. We estimate the diseases transmission rates by humans, $\beta$ , quarantine rate of exposed individuals, $\gamma_1$, isolation rate of infected individual, $\gamma_2$, rate at which quarantined individuals are isolated, $k_2$, recovery rate from quarantined individuals, $\sigma_1$, recovery rate from asymptomatic individuals, $\sigma_2$, recovery rate from isolated individuals, $\sigma_4$, and initial population sizes. The COVID-19 data are fitted using the optimization function 'fminsearchbnd' (MATLAB, R2017a). The estimated parameters are given in Table \ref{tab:mod1}. We also estimate the initial conditions of the human population and the estimated values are given by Table \ref{tab:mod3}. The fitting of the daily isolated COVID-19 cases in the UK are displayed in Figure \ref{Fig:new_case_fit}.

\begin{figure}[ht] 
\centering
	\includegraphics[width=0.45\textwidth]{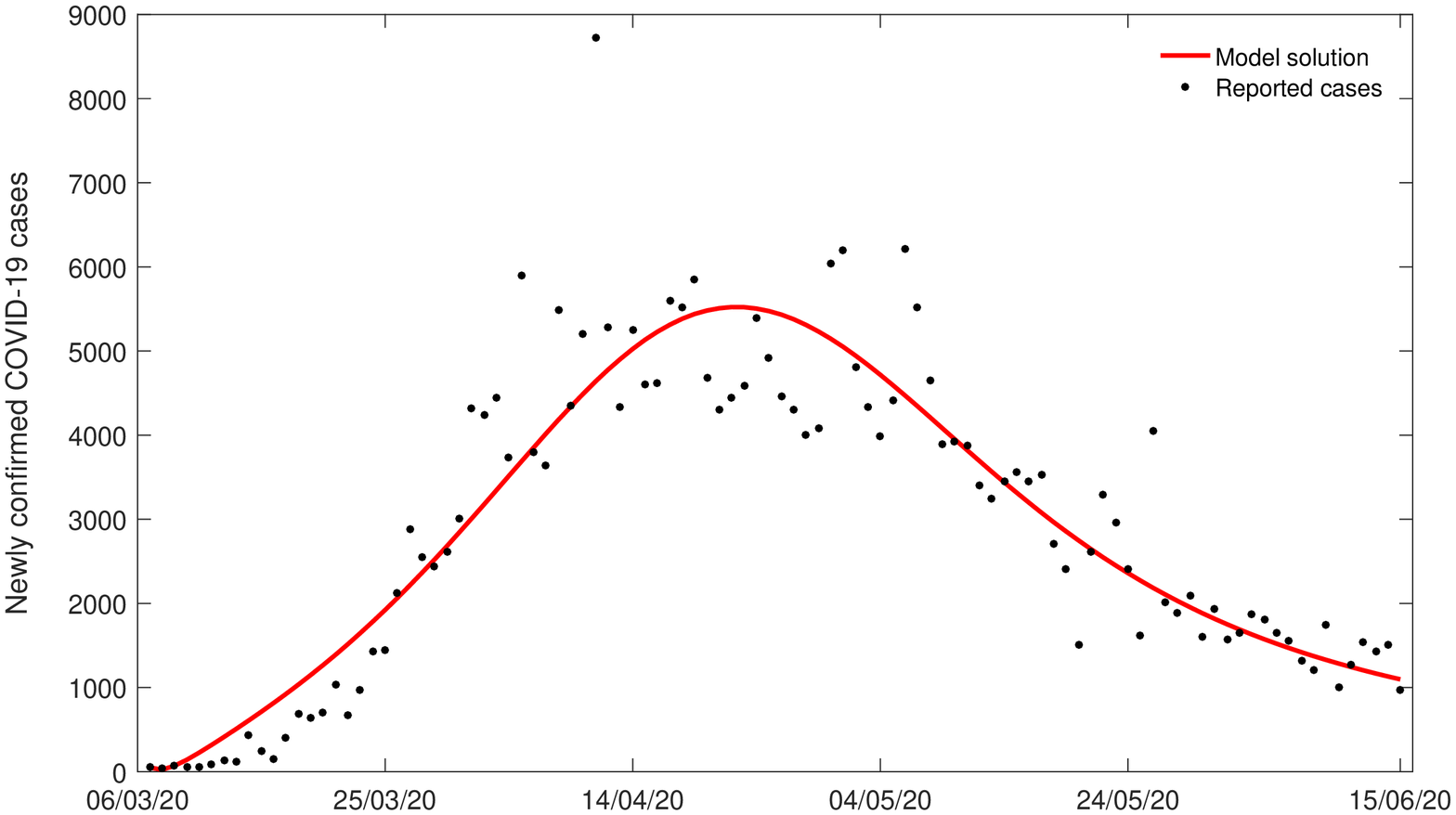}(a)
	\includegraphics[width=0.45\textwidth]{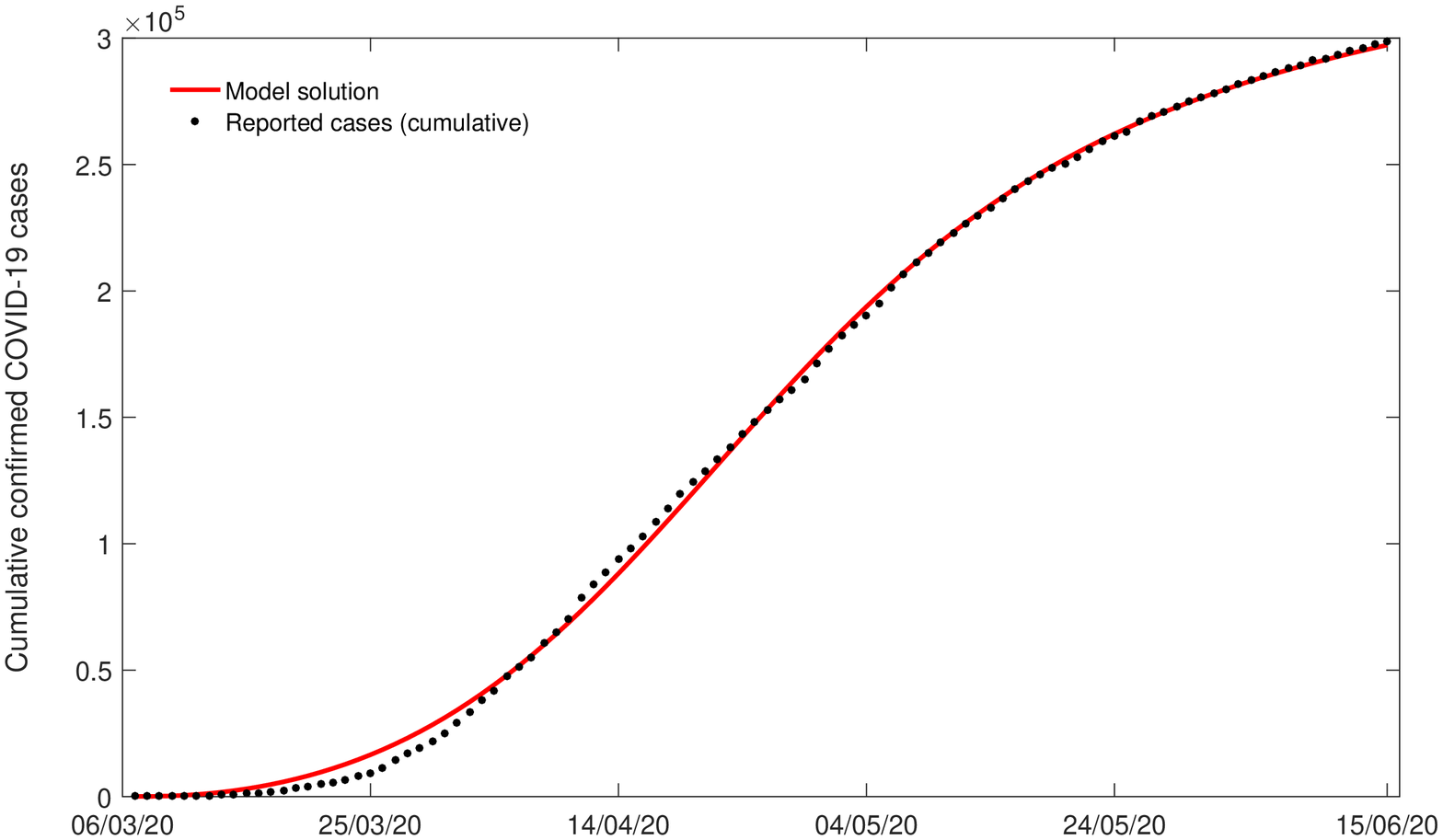}(b)
	\caption{(a) Model solutions fitted to daily new isolated COVID cases in the UK. (b) Model fitting with cumulative COVID-19 cases in the UK. Observed data points are shown in black circle and the solid red line depicts the model solutions.}
	\label{Fig:new_case_fit}
\end{figure}

\begin{table}[ht]
	\begin{center}
		\caption{Estimated initial population sizes for the UK.}
		\label{tab:mod3}
\begin{tabular}{c p{2cm} p{2cm}}

 \hline
Initial values & Value & Source \\
 \hline
$S(0)$ & 2000000 & Assumed \\
$E(0)$ & 103 & Estimated\\
$Q(0)$ & 0 & Assumed\\
$A(0)$ & 11016 & Estimated\\
$I(0)$ & 106 & Estimated \\
$J(0)$ & 48 & Data \\
$R(0)$ & 0 & Assumed \\
 \hline
\end{tabular}
 	\end{center}
\end{table}

Using these estimated parameters and the fixed parameters from Table \ref{tab:mod1}, we calculate the basic reproduction numbers ($R_0$) and control reproduction numbers ($R_c$) for the UK. The values for $R_0$ and $R_c$ are found to be 2.7048 and 2.3380 respectively. $R_c$ value is above unity, which indicates that they should increase the control interventions to limit future COVID-19 cases.  

\section{Short-term predictions}
In this section, the short-term prediction capability of the model \ref{EQ:eqn 2.1} is studied. Using parameters form Tables \ref{tab:mod1} and \ref{tab:mod3}, we simulate the newly isolated COVID-19 cases for the period 16 June, 2020 - 30 June, 2020 to check the accuracy of the predictions. Next, 10-day-ahead predictions are reported for the UK. The short-term prediction for the UK is depicted in Fig \ref{Fig:short_term_prediction}.

\begin{figure}[ht]
\includegraphics[width=1.0\textwidth]{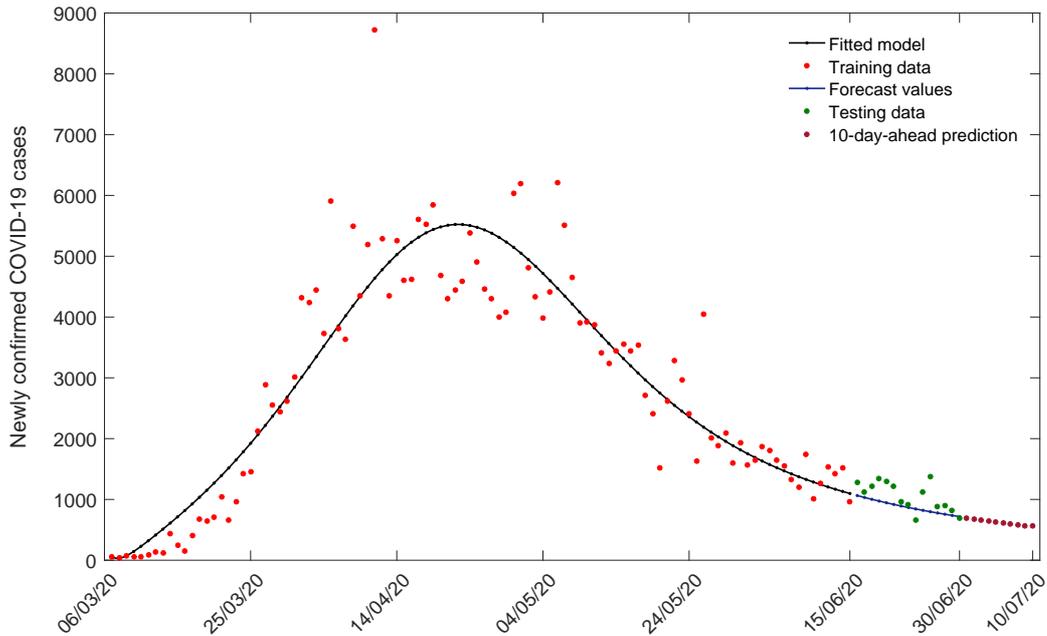}
\caption{Short term predictions for the UK. The blue line represent the predicted new isolated COVID cases while the solid dots are the actual cases.}
\label{Fig:short_term_prediction}
\end{figure}

We calculate two performance metrics, namely Mean Absolute Error (MAE) and Root Mean Square Error (RMSE) to assess the accuracy of the predictions. This is defined using a set of performance metrics as follows:\\
Mean Absolute Error (MAE):
\begin{align*}
     MAE=\frac{1}{N_p}\sum_{i=1}^{N_p} |Y(i)-\hat{Y}(i)|
\end{align*}
Root Mean Square Error (RMSE):
\begin{align*}
    RMSE=\sqrt{\frac{1}{N_p}\sum_{i=1}^{N_p} (Y(i)-\hat{Y}(i))^2}
\end{align*}
where $Y(i)$ represent original cases, $\hat{Y(i)}$ are predicted values and $N_p$ represents the sample size of the data. These performance metrics are found to be MAE=206.36 and RMSE=253.72. We found that the model performs excellently in case of the UK. The decreasing trend of newly isolated COVID-19 cases is also well captured by the model.

\section{Control strategies} 
In order to get an overview of most influential parameters, we compute the normalized sensitivity indices of the model parameters with respect to $R_c$. We have chosen parameters transmission rate between human population $\beta$, the control related parameters, $\gamma_1$, $\gamma_2$ and $k_2$, the recovery rates from quarantine individuals $\sigma_1$, asymptomatic individuals $\sigma_2$ and isolated individuals $\sigma_4$ and the effect of diseases induced mortality rate $\delta$ for sensitivity analysis. We compute normalized forward sensitivity indices of these parameters with respect to the control reproduction number $R_c$. We use the parameters from Table \ref{tab:mod1} and Table \ref{tab:mod3}. However, the mathematical definition of the normalized forward sensitivity index of a variable $m$ with respect to a parameter $\tau$ (where $m$ depends explicitly on the parameter $\tau$) is given as:

\begin{eqnarray*}
	X^\tau_m=\frac{\partial m}{\partial \tau}\times
	\frac{\tau}{m}.
\end{eqnarray*}

 The sensitivity indices of $R_c$ with respect to the parameters $\beta$, $\gamma_1$, $\gamma_2$, $k_2$, $\sigma_1$, $\sigma_2$, $\sigma_4$ and $\delta$ are given by Table \ref{tab:mod5}.

\begin{table}[ht]
	\begin{center}
		\caption{Normalized sensitivity indices of some parameters of the model \ref{EQ:eqn 2.1}}
		\label{tab:mod5}
\begin{tabular}{p{1.3cm} p{1.3cm} p{1.3cm} p{1.3cm} p{1.3cm} p{1.3cm} p{1.3cm} p{1.3cm}}

 \hline
 $X^{\beta}_{R_c}$ & $X^{\gamma_1}_{R_c}$ & $X^{\gamma_2}_{R_c}$ &  $X^{k_2}_{R_c}$ & $X^{\sigma_1}_{R_c}$ & $X^{\sigma_2}_{R_c}$ & $X^{\sigma_4}_{R_c}$ & $X^{\delta}_{R_c}$\\
 \hline
1.0000 & -0.1441 & -0.0268 & 0.0021 & -0.0879 & -0.4692 & -0.0757 & -0.0008\\
 \hline
\end{tabular}
 	\end{center}
\end{table}

The fact that $X^{\beta}_{R_c}=1$ means that if we increase 1\% in $\beta$, keeping other parameters be fixed, will produce $1$\% increase in $R_c$. Similarly, $X^{\sigma_2}_{R_c}=-0.4692$ means increasing the parameter $\sigma_2$ by $1$\%, the value of $R_c$ will be decrease by $0.4692$\% keeping the value of other parameters fixed. Therefore, the transmission rate between susceptible humans and COVID-19 infected humans is positively correlated and recovery rate from asymptomatic class is negatively correlated with respect to control reproduction number respectively. 

In addition, we draw the contour plots of $R_c$ with respect to the parameters $\gamma_1$ and $\gamma_2$ for the model \eqref{EQ:eqn 2.1} to investigate the effect of the control parameters on control reproduction number $R_c$, see Figure \ref{Fig:contour_plots1}.

\begin{figure}[ht]
\centering
	\includegraphics[width=0.45\textwidth]{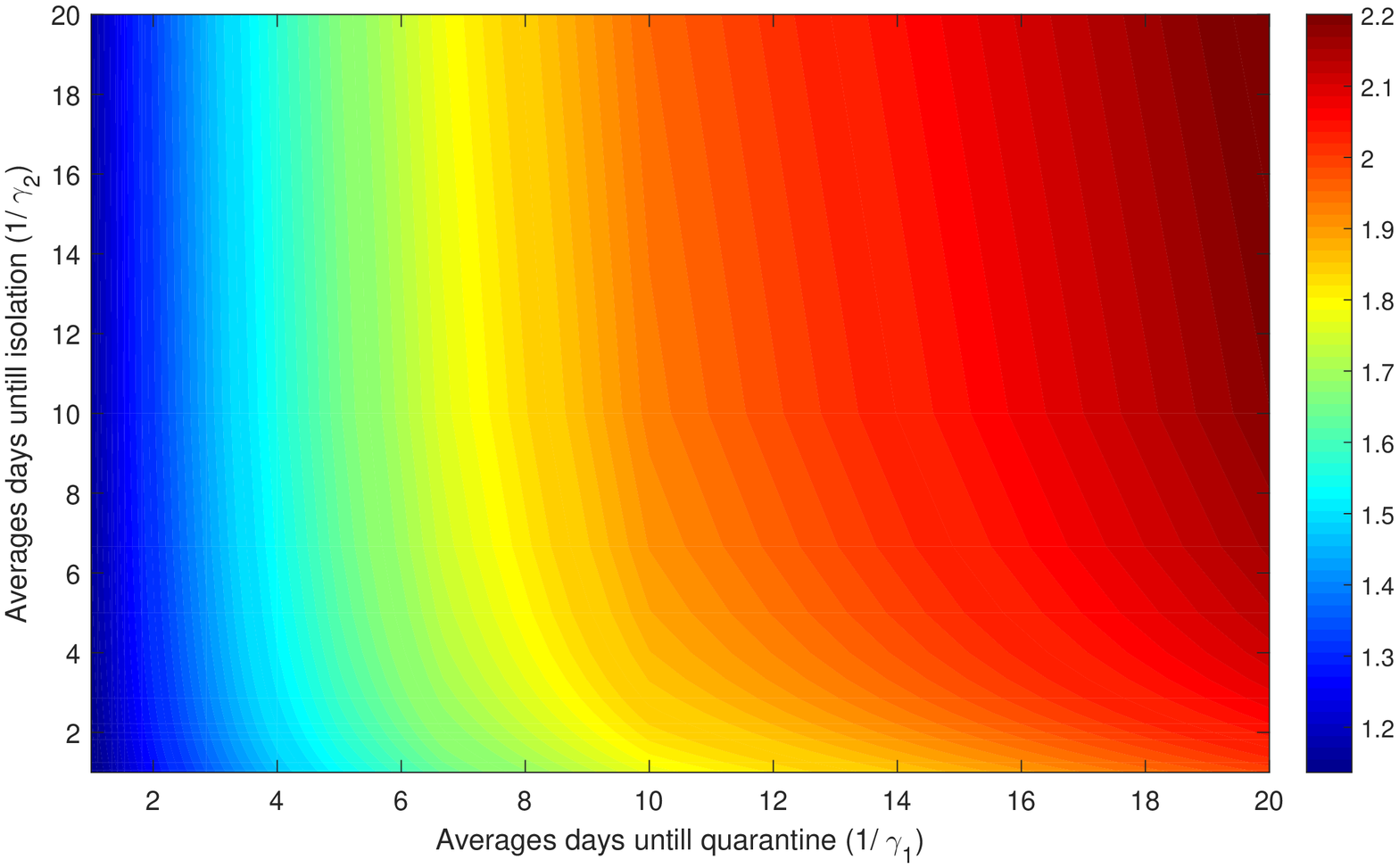}(a)
	\includegraphics[width=0.45\textwidth]{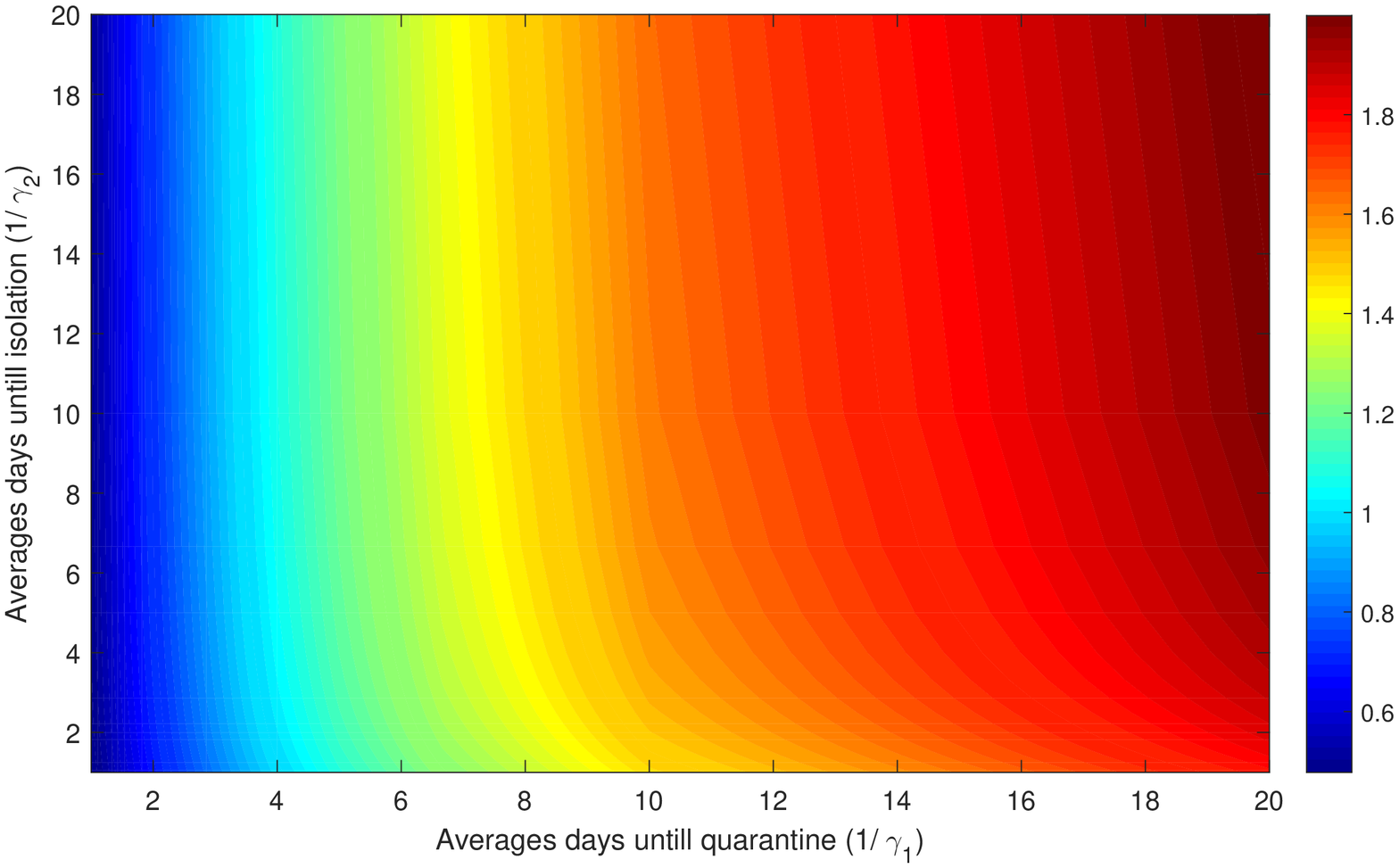}(b)
	\includegraphics[width=0.45\textwidth]{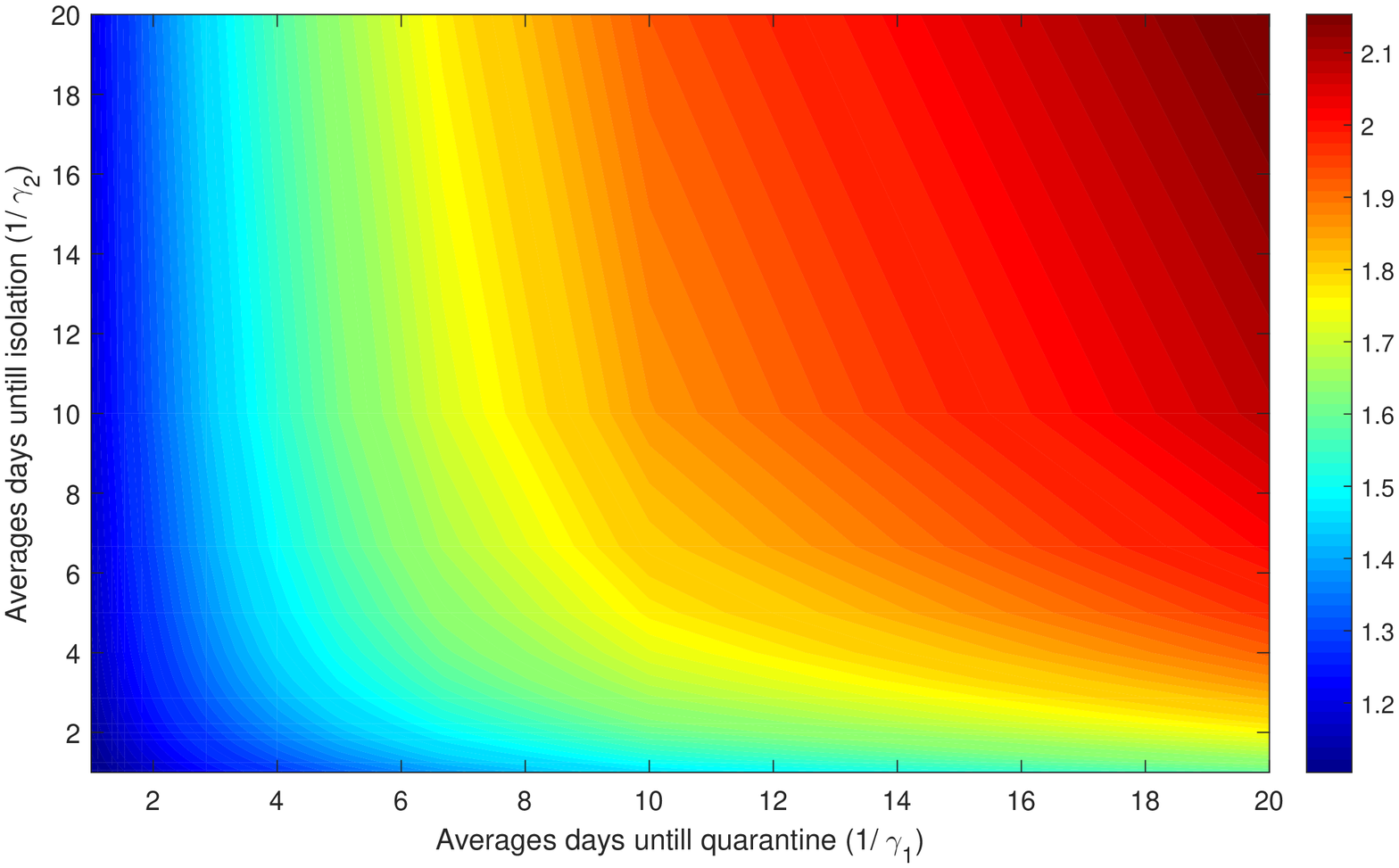}(c)
	\includegraphics[width=0.45\textwidth]{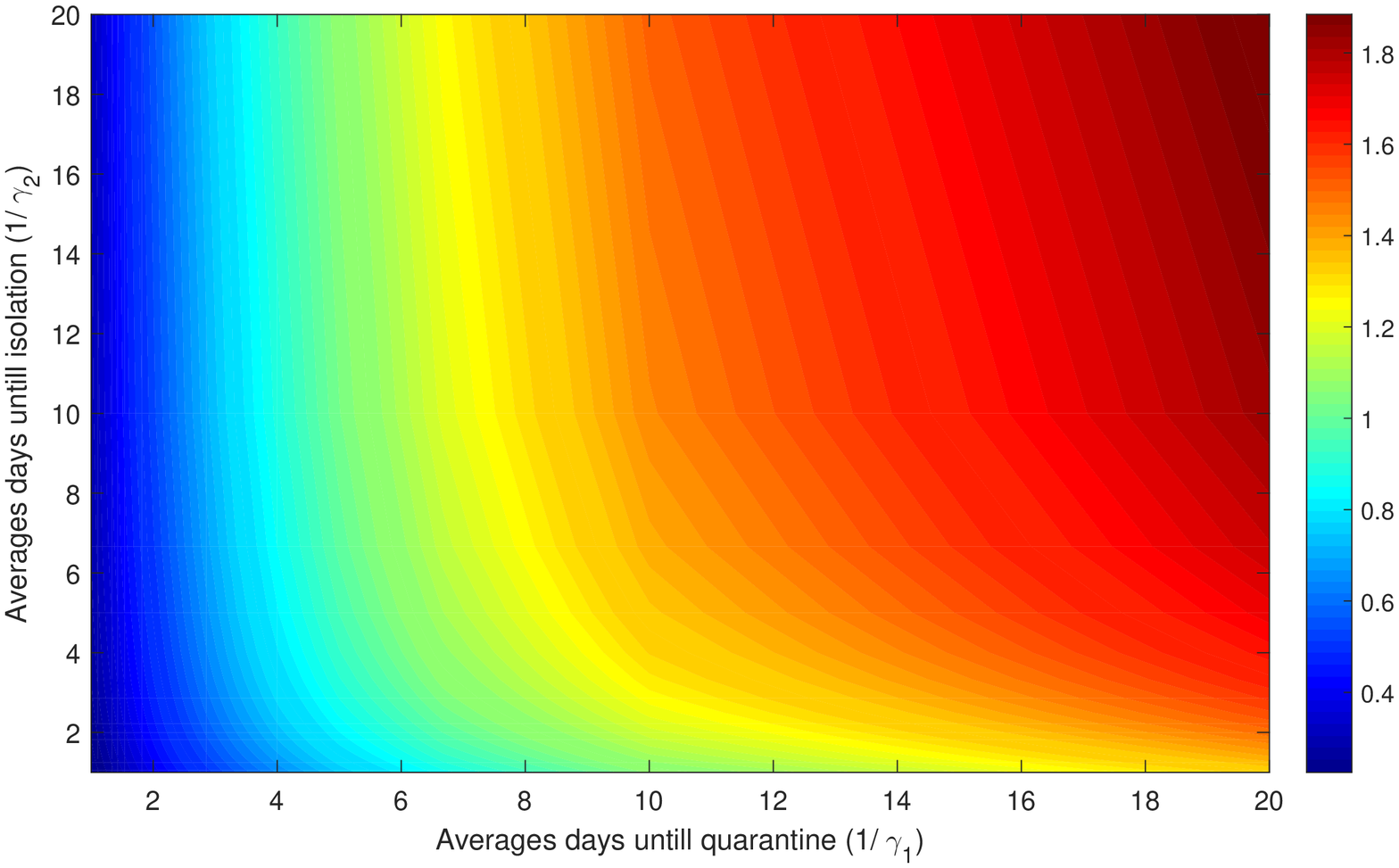}(d)
	\caption{Contour plots of $R_c$ versus average days to quarantine ($1/\gamma_1$) and isolation ($1/\gamma_2$) for the UK, (a) in the presence of both modification factors for quarantined ($r_Q$) and isolation ($r_J$); (b) in the presence of modification factors for isolation ($r_J$) only; (c) in the presence of modification factors for quarantined ($r_Q$) only and (d) in the absence of both modification factors for quarantined ($r_Q$) and isolation ($r_J$). All parameter values other than $\gamma_1$ and $\gamma_2$ are given in Table \ref{tab:mod1}.}
    \label{Fig:contour_plots1}
\end{figure}

The contour plots of Figure \ref{Fig:contour_plots1} show the dependence of $R_c$ on the quarantine rate $\gamma_1$ and the isolation rate $\gamma_2$ for the the UK. The axes of these plots are given as average days from exposed to quarantine ($1/\gamma_1$) and average days from starting of symptoms to isolation ($1/\gamma_2$). For both cases, the contours show that, increasing $\gamma_1$ and $\gamma_2$ reduces the amount of control reproduction number $R_c$ and, therefore, COVID cases. We find that quarantine and isolation are not sufficient to control the outbreak (see Figure \ref{Fig:contour_plots1}(a) and \ref{Fig:contour_plots1}(c)). With these parameter values, as $\gamma_1$ increases, $R_c$ decreases and similarly, when $\gamma_2$ increases, $R_c$ decreases. But, in the both cases $R_c>1$, and therefore the disease will persist in the population (i.e. the above control measures cannot lead to effective control of the epidemic). By contrast, our study shows that when the modification factor for quarantine become zero (so that $r_Q=0$), the outbreak can be controlled (see Figure \ref{Fig:contour_plots1}(b) and \ref{Fig:contour_plots1}(d)). From the above finding it follows that neither the quarantine of exposed individuals nor the isolation of symptomatic individuals will prevent the disease with the high value of the modification factor for quarantine. This control can be obtained by a significant reduction in COVID transmission during quarantine (that is reducing $r Q$ ). 

Furthermore, we study the effect of the parameters modification factor for quarantined individuals ($r_Q$), modification factor for isolated individuals ($r_J$) and transmission rate ($\beta$) on the cumulative new isolated COVID-19 cases ($J_{cum}$) in the UK. The cumulative number of isolated cases has been computed at day 100 (chosen arbitrarily). The effect of controllable parameters on ($J_{cum}$) are shown in Fig. \ref{fig:contour_plots2}.

\begin{figure}[ht]
    \centering
    \includegraphics[width=0.45\textwidth]{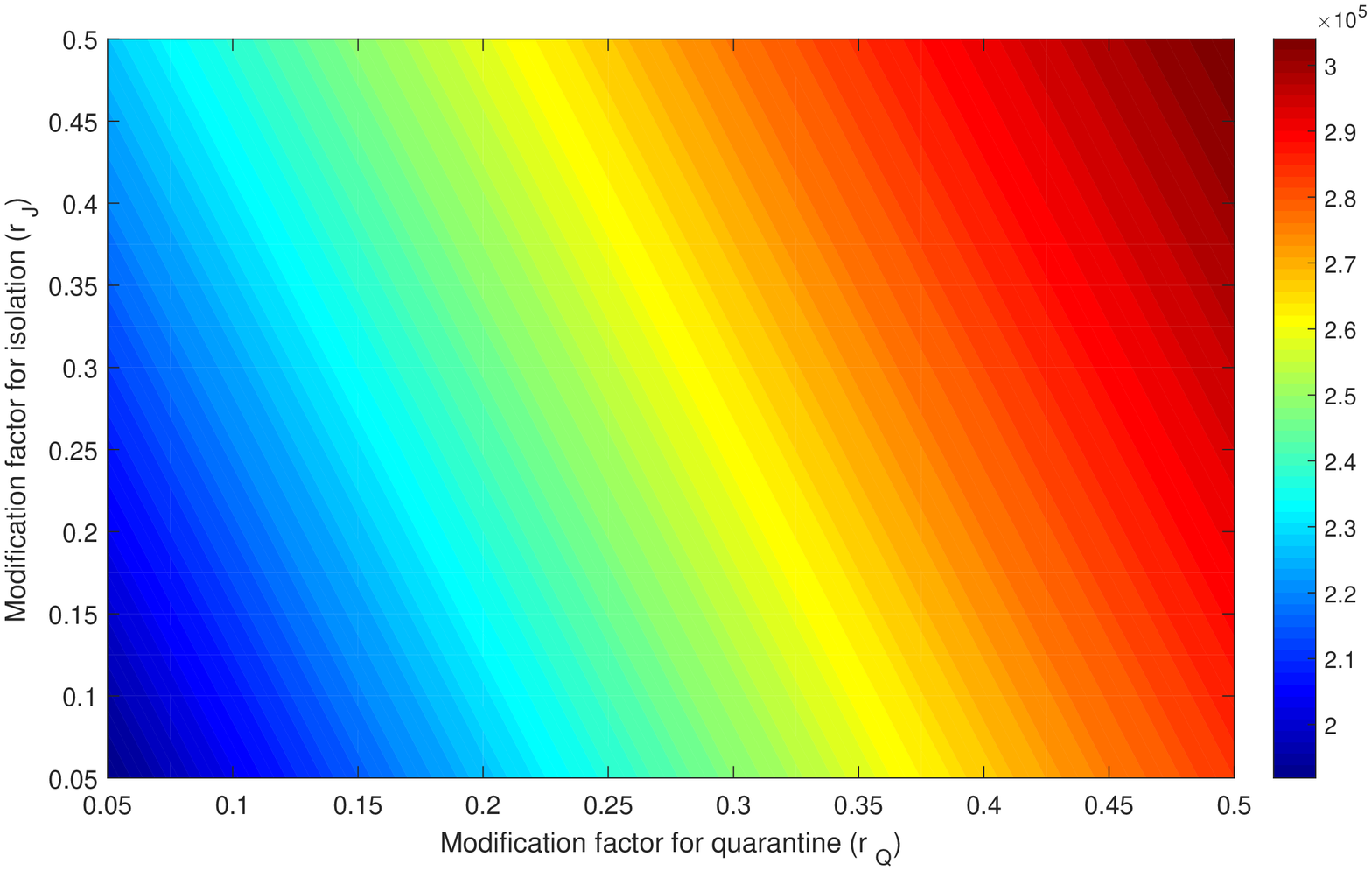}(a)
    \includegraphics[width=0.45\textwidth]{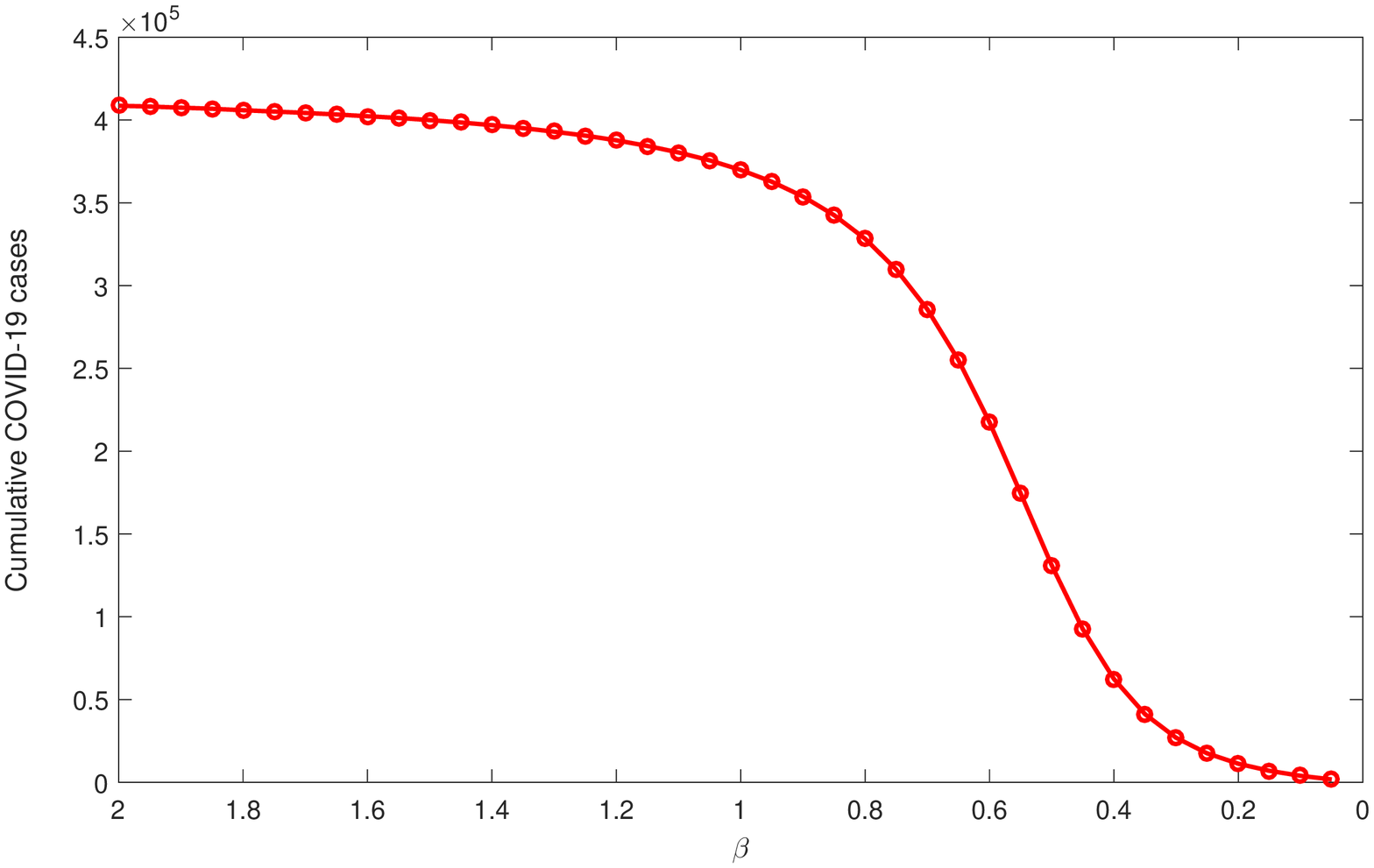}(b)
    \caption{Effect of controllable parameters $\gamma_1$, $\gamma_2$ and $\beta$ on the cumulative number of isolated COVID-19 cases. The left panel shows the variability of the $J_{cum}$ with respect to $\frac{1}{\gamma_1}$ and $\frac{1}{\gamma_2}$. The right panel shows $J_{cum}$ with decreasing transmission rate $\beta$.}
    \label{fig:contour_plots2}
\end{figure}

We observe that all the three parameters have significant effect on the cumulative outcome of the epidemic. From Fig. \ref{fig:contour_plots2}(a) it is clear that decrease in the modification factor for quarantined and isolated individuals will significantly reduce the value of $J_{cum}$. On the other hand Fig. \ref{fig:contour_plots2}(b) indicates, reduction in transmission rate will also slow down the epidemic significantly. These results point out that all the three control measures are quite effective in reduction of the COVID-19 cases in the UK. Thus, quarantine and isolation efficacy should be increased by means of proper hygiene and personal protection by health care stuffs. Additionally, the transmission coefficient can be reduced by avoiding contacts with suspected COVID-19 infected cases. 

Furthermore, We numerically calculated the thresholds $r_{\gamma_1}$ and $r_{\gamma_2}$ for the UK. The analytical expression of the thresholds are given in subsection ($3.5$). The effectiveness of quarantine and isolation depends on the values of the modification parameters $r_Q$ and $r_J$ for the reduction of infected individuals. The threshold value of $r_Q$ corresponding to quarantine parameter $\gamma_1$ is $r_{\gamma_1} = 0.9548$ and the threshold value of $r_J$ corresponding to isolation parameter $\gamma_2$ is $r_{\gamma_2} = 0.9861$.

\begin{figure}[ht]
\begin{center}
    \includegraphics[width=0.45\textwidth]{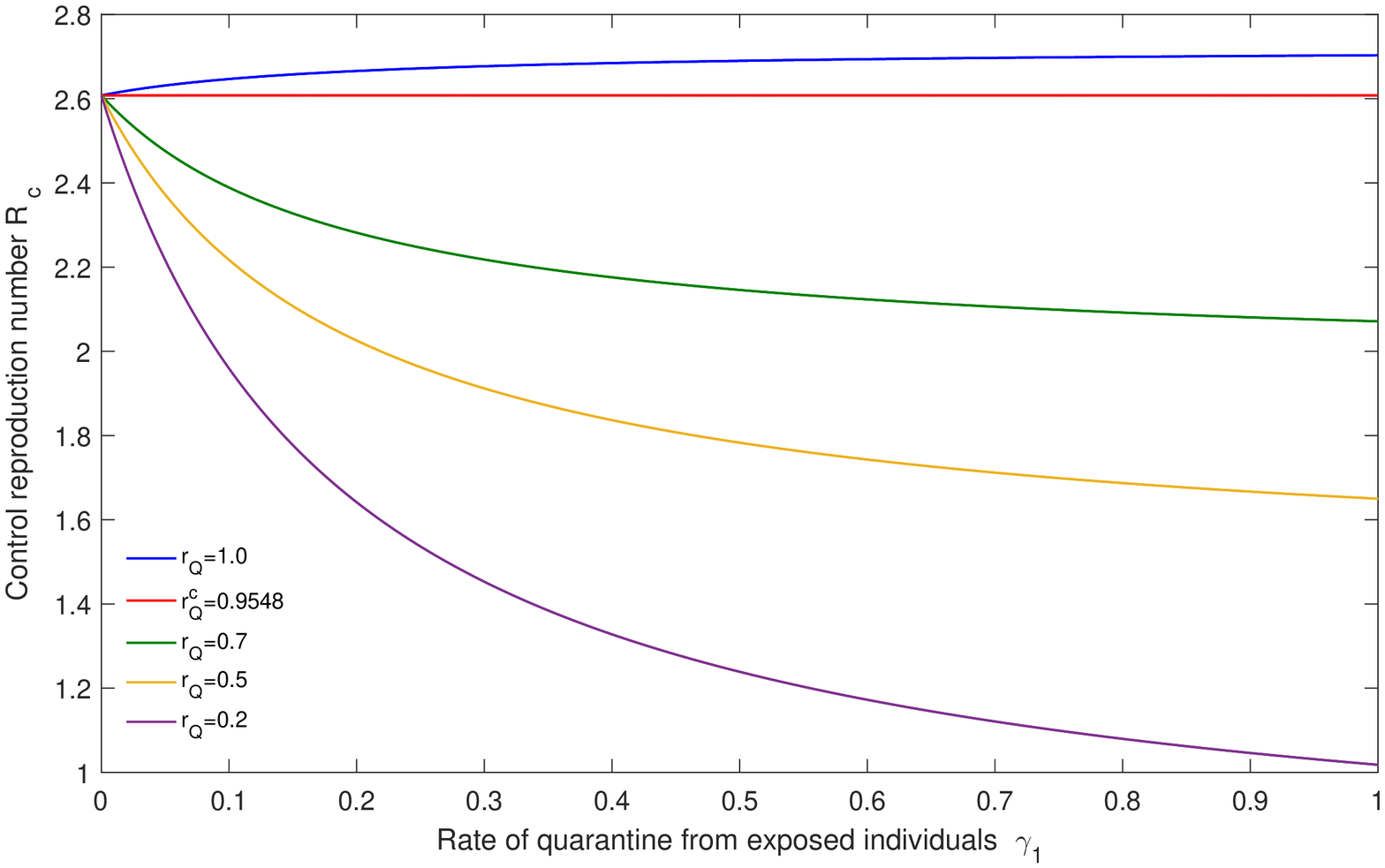}(a)
    \includegraphics[width=0.45\textwidth]{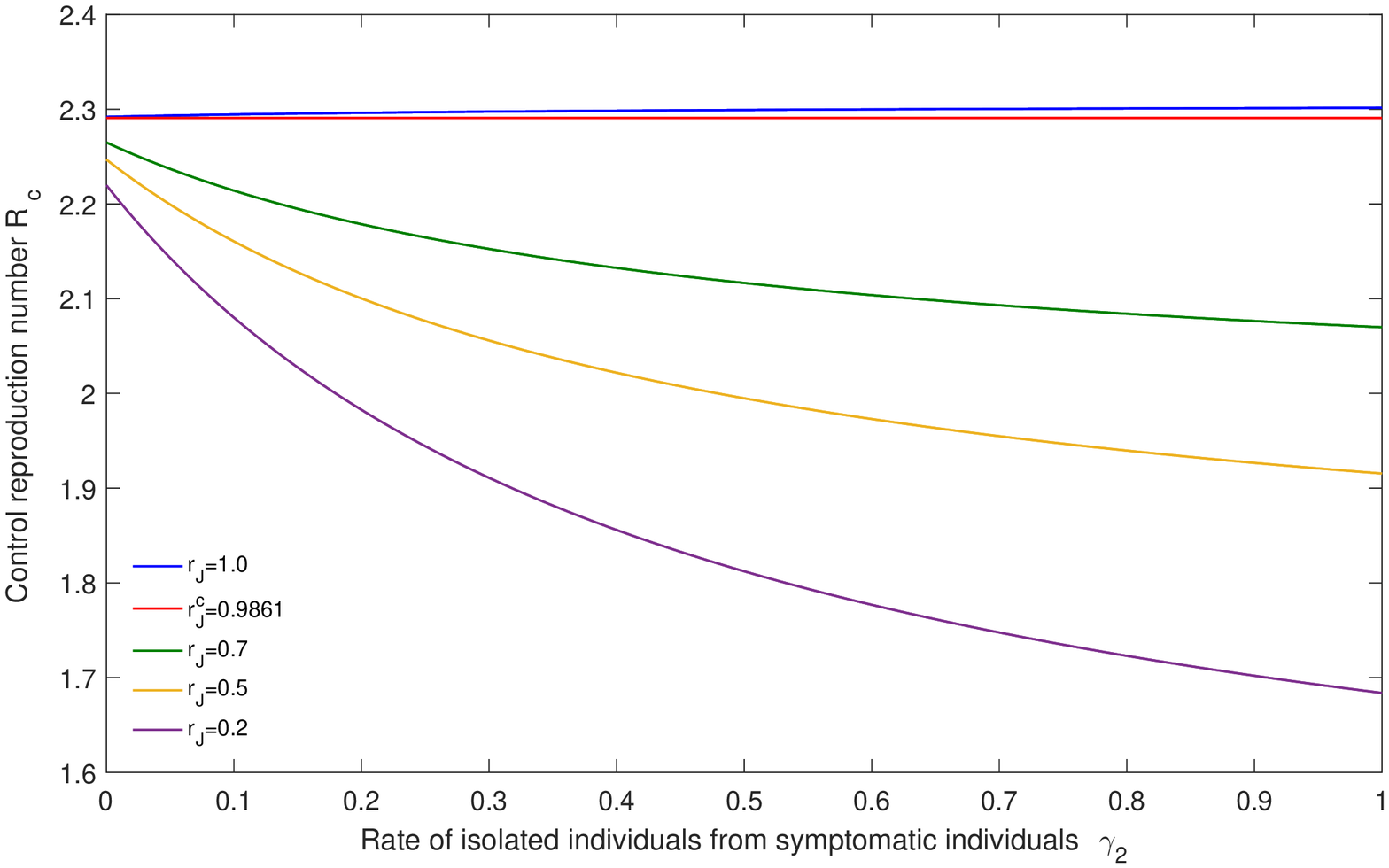}(b)
    \caption{Effect of isolation parameters $\gamma_1$ and $\gamma_2$ on control reproduction number $R_c$.}
    \label{fig:threshold_gamma1_gamma2}
    \end{center}
\end{figure}

From figure \ref{fig:threshold_gamma1_gamma2}(a) it is clear that quarantine parameter $\gamma_1$ has positive population-level impact ($R_c$ decreases with increase in $\gamma_1$) for $r_Q<0.9548$ and have negative population level impact for $r_Q>0.9548$. Similarly from the figure \ref{fig:threshold_gamma1_gamma2}(b), it is clear that, isolation has positive level impact for $r_J<0.9861$, whereas isolation has negative impact if $r_J>0.9861$. This result indicate that isolation and quarantine programs should run effective so that the modification parameters remain below the above mentioned threshold.

\section{Discussion}
During the period of an epidemic when human-to-human transmission is established and reported cases of COVID-19 are rising worldwide, forecasting is of utmost importance for health care planning and control the virus with limited resource. In this study, we have formulated and analyzed a compartmental epidemic model of COVID-19 to predict and control the outbreak. The basic reproduction number and control reproduction number are calculated for the proposed model. It is also shown that whenever $R_0<1$, the DFE of the model without control is globally asymptotically stable. The efficacy of quarantine of exposed individuals and isolation of infected symptomatic individuals depends on the size of the modification parameter to reduce the infectiousness of exposed ($r_Q$) and isolated ($r_J$) individuals. The usage of quarantine and isolation will have positive population-level impact if $r_Q < r_{\gamma_1}$ and $r_J<r_{\gamma_2}$ respectively. We calibrated the proposed model to fit daily data from the UK. Using the parameter estimates, we then found the basic and control reproduction numbers for the UK. Our findings suggest that independent self-sustaining human-to-human spread ($R_0>1$, $R_c>1$) is already present in the UK. The estimates of control reproduction number indicate that sustained control interventions are necessary to reduce the future COVID-19 cases. The health care agencies should focus on successful implementation of control mechanisms to reduce the burden of the disease. 

The calibrated model then checked for short-term predictability. It is seen that the model performs excellently (Fig. \ref{Fig:short_term_prediction}). The model predicted that the new cases in the UK will show decreasing trend in the near future. However, if the control measures are increased (or $R_c$ is decreased below unity to ensure GAS of the DFE) and maintained efficiently, the subsequent outbreaks can be controlled. 

Having an estimate of the parameters and prediction results, we performed control intervention related numerical experiments. Sensitivity analysis reveal that the transmission rate is positively correlated and quarantine and isolation rates negatively correlated with respect to control reproduction number. This indicate that increasing quarantine and isolation rates and decreasing transmission rate will decrease the control reproduction number and consequently will reduce the disease burden. 

While investigating the contour plots \ref{Fig:contour_plots1}, we found that effective management of quarantined individuals is more effective than management of isolated individuals to reduce the control reproduction number below unity. Thus if limited resources are available, then investing on the quarantined individuals will be more fruitful in terms of reduction of cases.

Finally, we studied the effect of modification factor for quarantined population, modification factor for isolated population and transmission rate on the newly infected symptomatic COVID-19 cases. Numerical results show that all the three control measures are quite effective in reduction of the COVID-19 cases in the UK (Fig. \ref{fig:contour_plots2}). The threshold analysis reinforce that the quarantine and isolation efficacy should be increased to reduce the epidemic (Fig. \ref{fig:threshold_gamma1_gamma2}). Thus, quarantine and isolation efficacy should be increased by means of proper hygiene and personal protection by health care stuffs. Additionally, the transmission coefficient can be reduced by avoiding contacts with suspected COVID-19 infected cases. 

In summary, our study suggests that COVID-19 has a potential to be endemic for quite a long period but it is controllable by social distancing measures and efficiency in quarantine and isolation. Moreover, if limited resources are available, then investing on the quarantined individuals will be more fruitful in terms of reduction of cases. The ongoing control interventions should be adequately funded and monitored by the health ministry. Health care officials should supply medications, protective masks and necessary human resources in the affected areas. 

\section*{Acknowledgements}
Sk Shahid Nadim receives senior research fellowship from CSIR, Government of India, New Delhi. Research of Indrajit Ghosh is financially supported by the Indian Statistical Institute, Kolkata through his visiting scientist position at this institute.

\bibliographystyle{plain}
\biboptions{square}
\bibliography{bib}

\begin{thebibliography}{10}

\bibitem{Who2019}
{WHO. Coronavirus} disease (covid-19) outbreak.
\newblock \url{https://www.who.int/emergencies
  /diseases/novel-coronavirus-2019}, 2019.
\newblock Retrieved : 2020-03-04.

\bibitem{chinadaily2019}
{Wuhan} wet market closes amid pneumonia outbreak.
\newblock
  \url{https://www.chinadaily.com.cn/a/202001/01/WS5e0c6a49a310cf3e35581e30.html},
  2019.
\newblock Retrieved : 2020-03-04.

\bibitem{cdcgov2020}
Centers for disease control and prevention: 2019 novel coronavirus.
\newblock \url{https: //www.cdc.gov/coronavirus/2019-ncov}, 2020.
\newblock Retrieved : 2020-03-10.

\bibitem{Worldometer2020}
{COVID-19} coronavirus outbreak.
\newblock \url{https://www.worldometers.info/coronavirus/#repro}, 2020.
\newblock Retrieved : 2020-03-04.

\bibitem{lifexp2018}
Life expectancy at birth, total (years) - china.
\newblock
  \url{https://data.worldbank.org/indicator/SP.DYN.LE00.IN?locations=CN}, 2020.
\newblock Retrieved : 2020-02-15.

\bibitem{nowcast2019}
{Nowcasting and Forecasting the Wuhan 2019-nCoV Outbreak.} available online:.
\newblock
  \url{https://files.sph.hku.hk/download/wuhan_exportation_preprint.pdf}, 2020.
\newblock Retrieved : 2020-03-04.

\bibitem{adameffectiveness}
J~Kucharski Adam, Klepac Petra, JK~Andrew, M~Kissler Stephen, L~Tang Maria, Fry
  Hannah, R~Julia, CMMID COVID-19 working group, et~al.
\newblock Effectiveness of isolation, testing, contact tracing, and physical
  distancing on reducing transmission of sars-cov-2 in different settings: A
  mathematical modelling study.
\newblock {\em The Lancet. Infectious diseases}, pages S1473--3099.

\bibitem{aldila2020mathematical}
Dipo Aldila, Sarbaz~HA Khoshnaw, Egi Safitri, Yusril~Rais Anwar, Aanisah~RQ
  Bakry, Brenda~M Samiadji, Demas~A Anugerah, M~Farhan~Alfarizi GH, Indri~D
  Ayulani, and Sheryl~N Salim.
\newblock A mathematical study on the spread of covid-19 considering social
  distancing and rapid assessment: The case of jakarta, indonesia.
\newblock {\em Chaos, Solitons \& Fractals}, page 110042, 2020.

\bibitem{anderson1991may}
Roy~M Anderson and M~Robert.
\newblock May. infectious diseases of humans: dynamics and control.
\newblock {\em Oxford Science Publications}, 36:118, 1991.

\bibitem{bogoch2020pneumonia}
Isaac~I Bogoch, Alexander Watts, Andrea Thomas-Bachli, Carmen Huber, Moritz~UG
  Kraemer, and Kamran Khan.
\newblock Pneumonia of unknown etiology in wuhan, china: Potential for
  international spread via commercial air travel.
\newblock {\em Journal of Travel Medicine}, 2020.

\bibitem{britton2020mathematical}
Tom Britton, Frank Ball, and Pieter Trapman.
\newblock A mathematical model reveals the influence of population
  heterogeneity on herd immunity to sars-cov-2.
\newblock {\em Science}, 2020.

\bibitem{castillo2002computation}
Carlos Castillo-Chavez, Zhilan Feng, and Wenzhang Huang.
\newblock On the computation of ro and its role on.
\newblock {\em Mathematical approaches for emerging and reemerging infectious
  diseases: an introduction}, 1:229, 2002.

\bibitem{castillo2004dynamical}
Carlos Castillo-Chavez and Baojun Song.
\newblock Dynamical models of tuberculosis and their applications.
\newblock {\em Mathematical Biosciences \& Engineering}, 1(2):361, 2004.

\bibitem{chakraborty2020real}
Tanujit Chakraborty and Indrajit Ghosh.
\newblock Real-time forecasts and risk assessment of novel coronavirus
  (covid-19) cases: A data-driven analysis.
\newblock {\em Chaos, Solitons \& Fractals}, page 109850, 2020.

\bibitem{chan2020familial}
Jasper Fuk-Woo Chan, Shuofeng Yuan, Kin-Hang Kok, Kelvin Kai-Wang To, Hin Chu,
  Jin Yang, Fanfan Xing, Jieling Liu, Cyril Chik-Yan Yip, Rosana Wing-Shan
  Poon, et~al.
\newblock A familial cluster of pneumonia associated with the 2019 novel
  coronavirus indicating person-to-person transmission: a study of a family
  cluster.
\newblock {\em The Lancet}, 395(10223):514--523, 2020.

\bibitem{cheng20202019}
Zhangkai~J Cheng and Jing Shan.
\newblock 2019 novel coronavirus: where we are and what we know.
\newblock {\em Infection}, pages 1--9, 2020.

\bibitem{chowell2009severe}
Gerardo Chowell, Stefano~M Bertozzi, M~Arantxa Colchero, Hugo Lopez-Gatell,
  Celia Alpuche-Aranda, Mauricio Hernandez, and Mark~A Miller.
\newblock Severe respiratory disease concurrent with the circulation of h1n1
  influenza.
\newblock {\em New England journal of medicine}, 361(7):674--679, 2009.

\bibitem{chowell2011characterizing}
Gerardo Chowell, Santiago Echevarr{\'\i}a-Zuno, Cecile Viboud, Lone Simonsen,
  James Tamerius, Mark~A Miller, and V{\'\i}ctor~H Borja-Aburto.
\newblock Characterizing the epidemiology of the 2009 influenza a/h1n1 pandemic
  in mexico.
\newblock {\em PLoS medicine}, 8(5), 2011.

\bibitem{clark2020global}
Andrew Clark, Mark Jit, Charlotte Warren-Gash, Bruce Guthrie, Harry~HX Wang,
  Stewart~W Mercer, Colin Sanderson, Martin McKee, Christopher Troeger,
  Kanyin~L Ong, et~al.
\newblock Global, regional, and national estimates of the population at
  increased risk of severe covid-19 due to underlying health conditions in
  2020: a modelling study.
\newblock {\em The Lancet Global Health}, 2020.

\bibitem{cowling2015preliminary}
Benjamin~J Cowling, Minah Park, Vicky~J Fang, Peng Wu, Gabriel~M Leung, and
  Joseph~T Wu.
\newblock Preliminary epidemiologic assessment of mers-cov outbreak in south
  korea, may--june 2015.
\newblock {\em Euro surveillance: bulletin Europeen sur les maladies
  transmissibles= European communicable disease bulletin}, 20(25), 2015.

\bibitem{davies2020effects}
Nicholas~G Davies, Adam~J Kucharski, Rosalind~M Eggo, Amy Gimma, W~John
  Edmunds, Thibaut Jombart, Kathleen O'Reilly, Akira Endo, Joel Hellewell,
  Emily~S Nightingale, et~al.
\newblock Effects of non-pharmaceutical interventions on covid-19 cases,
  deaths, and demand for hospital services in the uk: a modelling study.
\newblock {\em The Lancet Public Health}, 2020.

\bibitem{de2013commentary}
Raoul~J de~Groot, Susan~C Baker, Ralph~S Baric, Caroline~S Brown, Christian
  Drosten, Luis Enjuanes, Ron~AM Fouchier, Monica Galiano, Alexander~E
  Gorbalenya, Ziad~A Memish, et~al.
\newblock Commentary: Middle east respiratory syndrome coronavirus (mers-cov):
  announcement of the coronavirus study group.
\newblock {\em Journal of virology}, 87(14):7790--7792, 2013.

\bibitem{diekmann2000mathematical}
Odo Diekmann and Johan Andre~Peter Heesterbeek.
\newblock {\em Mathematical epidemiology of infectious diseases: model
  building, analysis and interpretation}, volume~5.
\newblock John Wiley \& Sons, 2000.

\bibitem{fraser2009pandemic}
Christophe Fraser, Christl~A Donnelly, Simon Cauchemez, William~P Hanage,
  Maria~D Van~Kerkhove, T~D{\'e}irdre Hollingsworth, Jamie Griffin, Rebecca~F
  Baggaley, Helen~E Jenkins, Emily~J Lyons, et~al.
\newblock Pandemic potential of a strain of influenza a (h1n1): early findings.
\newblock {\em science}, 324(5934):1557--1561, 2009.

\bibitem{gralinski2020return}
Lisa~E Gralinski and Vineet~D Menachery.
\newblock Return of the coronavirus: 2019-ncov.
\newblock {\em Viruses}, 12(2):135, 2020.

\bibitem{gumel2004modelling}
Abba~B Gumel, Shigui Ruan, Troy Day, James Watmough, Fred Brauer, P~Van~den
  Driessche, Dave Gabrielson, Chris Bowman, Murray~E Alexander, Sten Ardal,
  et~al.
\newblock Modelling strategies for controlling sars outbreaks.
\newblock {\em Proceedings of the Royal Society of London. Series B: Biological
  Sciences}, 271(1554):2223--2232, 2004.

\bibitem{heffernan2005perspectives}
Jane~M Heffernan, Robert~J Smith, and Lindi~M Wahl.
\newblock Perspectives on the basic reproductive ratio.
\newblock {\em Journal of the Royal Society Interface}, 2(4):281--293, 2005.

\bibitem{hethcote2000mathematics}
Herbert~W Hethcote.
\newblock The mathematics of infectious diseases.
\newblock {\em SIAM review}, 42(4):599--653, 2000.

\bibitem{huang2020clinical}
Chaolin Huang, Yeming Wang, Xingwang Li, Lili Ren, Jianping Zhao, Yi~Hu,
  Li~Zhang, Guohui Fan, Jiuyang Xu, Xiaoying Gu, et~al.
\newblock Clinical features of patients infected with 2019 novel coronavirus in
  wuhan, china.
\newblock {\em The Lancet}, 395(10223):497--506, 2020.

\bibitem{imai2020estimating}
Natsuko Imai, Ilaria Dorigatti, Anne Cori, Steven Riley, and Neil~M Ferguson.
\newblock Estimating the potential total number of novel coronavirus cases in
  wuhan city, china, 2020.

\bibitem{jit2020estimating}
Mark Jit, Thibaut Jombart, Emily~S Nightingale, Akira Endo, Sam Abbott, W~John
  Edmunds, et~al.
\newblock Estimating number of cases and spread of coronavirus disease
  (covid-19) using critical care admissions, united kingdom, february to march
  2020.
\newblock {\em Eurosurveillance}, 25(18):2000632, 2020.

\bibitem{kim2017middle}
KH~Kim, TE~Tandi, Jae~Wook Choi, JM~Moon, and MS~Kim.
\newblock Middle east respiratory syndrome coronavirus (mers-cov) outbreak in
  south korea, 2015: epidemiology, characteristics and public health
  implications.
\newblock {\em Journal of Hospital Infection}, 95(2):207--213, 2017.

\bibitem{kucharski2020early}
Adam~J Kucharski, Timothy~W Russell, Charlie Diamond, Yang Liu, John Edmunds,
  Sebastian Funk, Rosalind~M Eggo, Fiona Sun, Mark Jit, James~D Munday, et~al.
\newblock Early dynamics of transmission and control of covid-19: a
  mathematical modelling study.
\newblock {\em The lancet infectious diseases}, 2020.

\bibitem{kwok2019epidemic}
Kin~On Kwok, Arthur Tang, Vivian~WI Wei, Woo~Hyun Park, Eng~Kiong Yeoh, and
  Steven Riley.
\newblock Epidemic models of contact tracing: Systematic review of transmission
  studies of severe acute respiratory syndrome and middle east respiratory
  syndrome.
\newblock {\em Computational and structural biotechnology journal}, 2019.

\bibitem{lai2020assessing}
Shengjie Lai, Isaac Bogoch, Nick Ruktanonchai, Alexander Watts, Yu~Li, Jianzing
  Yu, Xin Lv, Weizhong Yang, Hongjie Yu, Kamran Khan, et~al.
\newblock Assessing spread risk of wuhan novel coronavirus within and beyond
  china, january-april 2020: a travel network-based modelling study.
\newblock {\em medRxiv}, 2020.

\bibitem{li2003angiotensin}
Wenhui Li, Michael~J Moore, Natalya Vasilieva, Jianhua Sui, Swee~Kee Wong,
  Michael~A Berne, Mohan Somasundaran, John~L Sullivan, Katherine Luzuriaga,
  Thomas~C Greenough, et~al.
\newblock Angiotensin-converting enzyme 2 is a functional receptor for the sars
  coronavirus.
\newblock {\em Nature}, 426(6965):450--454, 2003.

\bibitem{lipsitch2003transmission}
Marc Lipsitch, Ted Cohen, Ben Cooper, James~M Robins, Stefan Ma, Lyn James,
  Gowri Gopalakrishna, Suok~Kai Chew, Chorh~Chuan Tan, Matthew~H Samore, et~al.
\newblock Transmission dynamics and control of severe acute respiratory
  syndrome.
\newblock {\em Science}, 300(5627):1966--1970, 2003.

\bibitem{may1991infectious}
Robert~M May.
\newblock {\em Infectious diseases of humans: dynamics and control}.
\newblock Oxford University Press, 1991.

\bibitem{muniz2020epidemic}
Kamalich Muniz-Rodriguez, Gerardo Chowell, Chi-Hin Cheung, Dongyu Jia, Po-Ying
  Lai, Yiseul Lee, Manyun Liu, Sylvia~K Ofori, Kimberlyn~M Roosa, Lone
  Simonsen, et~al.
\newblock Epidemic doubling time of the covid-19 epidemic by chinese province.
\newblock {\em medRxiv}, 2020.

\bibitem{rajagopal2020fractional}
Karthikeyan Rajagopal, Navid Hasanzadeh, Fatemeh Parastesh, Ibrahim~Ismael
  Hamarash, Sajad Jafari, and Iqtadar Hussain.
\newblock A fractional-order model for the novel coronavirus (covid-19)
  outbreak.
\newblock {\em Nonlinear Dynamics}, pages 1--8, 2020.

\bibitem{sardar2020realistic}
Tridip Sardar, Indrajit Ghosh, Xavier Rod{\'o}, and Joydev Chattopadhyay.
\newblock A realistic two-strain model for mers-cov infection uncovers the high
  risk for epidemic propagation.
\newblock {\em PLoS neglected tropical diseases}, 14(2):e0008065, 2020.

\bibitem{sardar2020assessment}
Tridip Sardar, Sk~Shahid Nadim, and Joydev Chattopadhyay.
\newblock Assessment of 21 days lockdown effect in some states and overall
  india: a predictive mathematical study on covid-19 outbreak.
\newblock {\em arXiv preprint arXiv:2004.03487}, 2020.

\bibitem{tang2020updated}
Biao Tang, Nicola~Luigi Bragazzi, Qian Li, Sanyi Tang, Yanni Xiao, and Jianhong
  Wu.
\newblock An updated estimation of the risk of transmission of the novel
  coronavirus (2019-ncov).
\newblock {\em Infectious Disease Modelling}, 2020.

\bibitem{tang2020estimation}
Biao Tang, Xia Wang, Qian Li, Nicola~Luigi Bragazzi, Sanyi Tang, Yanni Xiao,
  and Jianhong Wu.
\newblock Estimation of the transmission risk of the 2019-ncov and its
  implication for public health interventions.
\newblock {\em Journal of Clinical Medicine}, 9(2):462, 2020.

\bibitem{van2008further}
P~Van~den Driessche and James Watmough.
\newblock Further notes on the basic reproduction number.
\newblock In {\em Mathematical epidemiology}, pages 159--178. Springer, 2008.

\bibitem{van2002reproduction}
Pauline Van~den Driessche and James Watmough.
\newblock Reproduction numbers and sub-threshold endemic equilibria for
  compartmental models of disease transmission.
\newblock {\em Mathematical biosciences}, 180(1-2):29--48, 2002.

\bibitem{yang1996permanence}
Xia Yang, Lansun Chen, and Jufang Chen.
\newblock Permanence and positive periodic solution for the single-species
  nonautonomous delay diffusive models.
\newblock {\em Computers \& Mathematics with Applications}, 32(4):109--116,
  1996.

\end{thebibliography}

\end{document}